\documentclass{article}

\usepackage[top=1in, bottom=1in, left=1in, right=1in]{geometry}
\usepackage[utf8]{inputenc} 
\usepackage[T1]{fontenc}    
\usepackage{hyperref}       
\usepackage{url}            
\usepackage{booktabs}       
\usepackage{amsfonts}       
\usepackage{nicefrac}       
\usepackage{xcolor}         
\usepackage{bbm}
\usepackage{amsmath}
\usepackage{mathtools}
\usepackage{graphicx}
\usepackage{subcaption}
\usepackage{float}
\usepackage{amssymb}
\usepackage{amsthm}
\usepackage{multirow}
\usepackage{adjustbox}
\usepackage[title]{appendix}
\usepackage[ruled]{algorithm2e}
\usepackage{tikz}
\usepackage{placeins}
\usepackage{natbib}
\usepackage{setspace}
\usepackage{pgfplots}
\usetikzlibrary{patterns, shapes.geometric, positioning, arrows, arrows.meta, shadows, matrix, fit, backgrounds}
\pgfplotsset{compat=1.11}

\newcommand{\Halmos}{}
\newtheorem{assumption}{Assumption}
\newtheorem{definition}{Definition}
\newtheorem{lemma}{Lemma}

\newtheorem{theorem}{Theorem}
\usepackage{thmtools,thm-restate}

\DeclareMathOperator{\ind}{\mathbbm{1}}
\DeclareMathOperator{\supp}{supp}
\DeclareMathOperator{\comp}{ComP}
\DeclareMathOperator{\alwayst}{At}
\DeclareMathOperator{\nevert}{Nt}
\DeclareMathOperator{\defier}{DeF}
\DeclareMathOperator{\indes}{Ind}
\DeclareMathOperator{\cat}{Cat}
\newcommand{\abbreviation}[2]{\textit{#2 (#1)}}
\newcommand{\sprodk}[2]{\left\langle #1 , \, #2 \right\rangle} 
\newcommand{\hldef}[1]{\textit{#1}}
\newcommand{\abbrev}[2]{#1 \textit{(#2)}}
\newcommand{\defeq}{\coloneqq} 
\newcommand{\Rn}{\mathbb{R}} 
\newcommand{\Nn}{\mathbb{N}} 
\usepackage[font={footnotesize}]{caption}

\newcommand{\KEYWORDS}[1]{\textbf{Keywords:} #1}

\title{Effect Identification and Unit Categorization in the Multi-Score Regression Discontinuity Design with Application to LED Manufacturing}

\author{
  Philipp Alexander Schwarz\footnote{The authors contributed equally.} \\
  AMS Osram \\
  University of Hamburg \\
  \texttt{Philipp.Schwarz@ams-osram.com} \\
  \and
  Oliver Schacht\footnotemark[1] \\
  University of Hamburg \\
  \texttt{oliver.schacht@uni-hamburg.de} \\
  \and
  Sven Klaassen \\
  University of Hamburg \\
  Economic AI\\
  \and
  Johannes Oberpriller \\
  AMS Osram \\
  \and
  Martin Spindler \\
  University of Hamburg \\
  Economic AI \\
}
\date{}

\begin{document}

\maketitle

\abstract{%
\abbreviation{Regression discontinuity design}{RDD} is a widely used framework for identifying and estimating
causal effects at the cutoff of a single running variable.
In practice, however, decision-making often involves multiple thresholds and criteria, especially in production systems.
Standard \abbreviation{multi-score RDD}{MRD} methods address this complexity by reducing the problem to a one-dimensional design.
This simplification allows existing approaches to be used to identify and estimate causal effects,
but it can introduce non-compliance by misclassifying units relative to the original cutoff rules.
We develop theoretical tools to detect and reduce ``fuzziness'' 
when estimating the cutoff effect for units that comply with individual subrules of a multi-rule system.
In particular, we propose a formal definition and categorization of unit behavior types under multi-dimensional cutoff rules,
extending standard classifications of compliers, alwaystakers, and nevertakers, and incorporating defiers and indecisive units.
We further identify conditions under which cutoff effects for compliers can be estimated in multiple dimensions,
and establish when identification remains valid after excluding nevertakers and alwaystakers.
In addition, we examine how decomposing complex Boolean cutoff rules (such as AND- and OR-type rules) into simpler components
affects the classification of units into behavioral types and improves estimation 
by making it possible to identify and remove non-compliant units more accurately.
We validate our framework using both semi-synthetic simulations calibrated to production data and real-world data from opto-electronic semiconductor manufacturing.
The empirical results demonstrate that our approach has practical value in refining production policies and reduces estimation variance.
This underscores the usefulness of the MRD framework in manufacturing contexts.
}%

\KEYWORDS{Causal Inference, Machine Learning, Regression Discontinuity, Causal Effect Identification, Multi-Score RDD, Rework, Production Optimization, Phosphor-converted white LEDs}

\doublespacing

\section{Introduction}
In many areas of economics and management, decision-making depends on multiple criteria. Examples include targeted marketing, education policies, and political advertising. The complexity of such settings has motivated the development of formal approaches such as multi-criteria decision analysis, which make trade-offs explicit and offer methods for comparing and ranking alternatives. Once such decision rules are applied in practice, a further challenge is to assess their causal impact. The focus then shifts from choosing among options to evaluating the consequences of those choices. Causal inference, and in particular regression discontinuity design (RDD), provides tools for this purpose.

RDD is a well-established quasi-experimental approach for identifying causal effects whenever treatment assignment changes discontinuously at a threshold in an observed running variable. It has been applied in diverse fields, including economics  \citep{Hartmann2011, Card2015, Flammer2015, Calvo2019},
public policy \citep{Lee2008} and the social sciences \citep{Angrist1999}.
The appeal of RDD lies in its ability to produce credible causal estimates without relying on the assumptions that are required for most causal inference approaches. In particular, it does not require \textit{unconfoundedness} \cite{Rubin1974} or \textit{global positivity} \citep{Austin2011}, but instead depends on continuity of potential outcomes and the absence of precise manipulation around the cutoff.

The classic RDD exploits a discontinuity in a single treatment assignment rule that depends on an observed running or score variable. When correctly specified, this setting enables average treatment effects to be identified for units near the cutoff. Under mild conditions, it further allows counterfactual reasoning about the decision rule in a small neighborhood around the boundary  \citep{Cattaneo2019, Imbens2008, Lee2010}.

As noted above, treatment decisions in many real-world contexts are based on multiple criteria, particularly in industrial, operational, and engineering systems \citep{Sabaei2015}.
The natural extension of the classical RDD is multi-score RDD (MRD). In MRD, treatment depends on a joint rule requiring that multiple score variables satisfy specified threshold conditions simultaneously. The idea was first proposed by  \citet{papay2011}. Influential surveys that develop estimators and examine them empirically include  \citet{Reardon2012}, \citet{wong2013} and \citet{porter2017}. MRD has been widely applied, especially in test-score-based settings \citep{an2024} and geographical applications \citep{keele2015}. However, strategies for identifying MRD effects when multiple cutoffs are combined through arbitrary Boolean rules are still underdeveloped in the literature.

This paper addresses the identification and estimation challenges that arise in multi-score RDD settings. We develop new results tailored to these contexts. In particular, we introduce a formal framework that defines and categorizes behavioral types (``compliers'', ``alwaystakers'', ``nevertakers'', ``defiers'') and analyzes how these types transform when treatment assignment rules change. We investigate general Boolean cutoff mechanisms, such as ``AND-type'' and ``OR-type'' rules, which are commonly observed in MRD applications, and analyze how their properties influence local identification. Further, we demonstrate how the complier effect can be identified using these categories. Finally, we illustrate our results with data from a light-emitting diode (LED) production facility, in which each batch is evaluated based on several metrics, such as mean color point measurements or the proportion of units within a lot that meet specified quality requirements.

Our application of the proposed strategies to data from the LED manufacturing industry demonstrates the usefulness of the framework in a real-world setting where treatment assignment depends on multiple quality indicators. We combine the framework with recently developed estimators that use machine learning adjustments to further reduce variance in estimation \citep{noack2024}. Estimating the causal effect of threshold-based production decisions provides guidance for tuning decision thresholds to improve manufacturing outcomes. In practice, managers need to understand how marginal adjustments to these thresholds affect outcomes such as yield rates and thus overall production efficiency. We complement the empirical analysis with a simulated study based on a semi-synthetic replica of the production environment, highlighting the value of multi-score RDD for counterfactual analysis and manufacturing policy design.

This study contributes to the operations research literature by linking advances in causal machine learning with practical decision-making in manufacturing systems. Our theoretical results on effect identification in MRD create new opportunities for its application in decision-making contexts. By accommodating complex, multi-criteria assignment mechanisms, the framework reflects the realities of operational decision-making practice and provides tools for more informed and effective policy design.

The paper is organized as follows. Section 2 reviews related literature on RDD and its extensions to multi-score settings. Section 3 presents our theoretical framework for identification under general Boolean assignment rules. Section 4 describes the empirical application in opto-electronic semiconductor manufacturing. Section 5 reports results from synthetic and semi-synthetic environments. Section 6 discusses findings from the real-data application, emphasizing the implications for threshold calibration and production efficiency. Section 7 concludes.

\section{Related Literature}\label{sec:prev_work}
RDD dates back to a study on scholarship programs by  \citet{thistlethwaite1960}. In the following decades, the approach gained popularity in the social sciences, with notable applications by  \cite{Angrist1999} and \cite{Black1999}. Building on this early empirical work, \citet{Hahn2001} provided the first formalized treatment of RDD within the potential outcomes framework and established non-parametric identification and estimation results for both sharp and fuzzy designs, firmly establishing RDD as a core method in modern econometrics.

Important surveys, including those by  \citet{Imbens2008}, \citet{Lee2010} and \citet{Cattaneo2019}, synthesize guidance on estimation choices in RDD, such as bandwidth selection, polynomial order, and diagnostic checks. More recently, several theoretical advances have refined the method. \citet{Calonico2014} developed bias-corrected point estimators and robust confidence intervals that improve coverage accuracy and have become the empirical default for RDD inference. \cite{calonico2019} examined the inclusion of covariates in local regression to reduce variance, while \citet{Calonico2019b} analyzed optimal bandwidth selection for robust estimators. Further extensions incorporate machine-learning-based covariate adjustments to further reduce variance, integrating RDD into the broader toolbox of causal machine learning \citep{Kreiss2022, Arai2025, noack2024}.

The idea of extending RDD to multi-score settings dates back to \citet{papay2011}, who showed that causal effects can be identified along a two-dimensional cutoff frontier when program eligibility depends jointly on multiple scores. \citet{Reardon2012} systematically collected and formalized MRD identification strategies for cases in which running variables jointly cross their cutoffs to form a multidimensional frontier, corresponding to an ``AND''-type rule. They introduced frontier-projection and distance-based estimators, and highlighted the steep data-density and bandwidth challenges that arise as the dimensionality of the assignment space grows. \citet{wong2013} advanced this line of work with extensive simulation studies, demonstrating that estimators that weight observations along the entire cutoff frontier achieve lower bias and higher efficiency than naïve one-dimension-at-a-time approaches. Their findings underscore the importance of exploiting the joint assignment rule in practice. In another large-scale study, \cite{porter2017}  reached similar conclusions and provided practical recommendations for MRD applications. A related special case, geographic RDD, was introduced by \citet{keele2015} to evaluate policies along irregular geographic borders. 

More recent research has focused on theoretical advances. \citet{Choi2018} allowed for partial treatment effects when only some running variables cross their cutoffs and introduced thin-plate spline estimators to flexibly recover heterogeneous impacts along a high-dimensional frontier. Building upon earlier work by \cite{Imbens2009}, \cite{Choi2023} analyzed two-dimensional MRD under limited compliance and examined the consequences for estimating complier effects in fuzzy MRD.

Research has also produced more sophisticated MRD estimators. \citep{Imbens2019} proposed a minimax approach, \citep{liu2024} introduced decision-tree methods, and \citep{Sawada2025} developed fully multivariate local-polynomial estimators. \citet{Cattaneo2025} examined the statistical properties of important MRD estimators in detail and provided practical recommendations.

RDD has additionally been applied in a variety of business and management contexts. \citep{Hnermund2021} studied its use in business decision-making, while \citet{Ho2017, Mithas2022} examined applications in production and operations management. \citet{Calvo2019} used a one-dimensional fuzzy RDD to estimate the effect of procurement officers on budget overruns and delays in public infrastructure projects. \citet{Leung2020} analyzed the effects of labor unionization on the operating performance of supplier firms. In the automobile sector, \citet{Hu2021} employed RDD to assess the impact of regulatory changes on compliance. To the best of our knowledge, however, neither RDD nor MRD has yet been directly applied to production policies.

This paper contributes to the literature in two main ways. First, we derive new theoretical results on unit categorization and effect identification in MRD, which are particularly useful in settings with complex, multi-dimensional decision rules. Secondly, we apply these results to real data, demonstrating the applicability of MRD in operations policy-making and providing evidence to guide actual production policy.

\section{General Identification Strategies in Multi-score RDD}
\label{sec:theory}
In this section, we first review the basic elements of the RDD
and its extension to the multi-score case. 
We then introduce a formal definition of common unit types (e.g., compliers, defiers)
in multi-score, two-stage decision settings that employ cutoff rules for treatment assignment.
Using this categorization, we derive an identification result and show that,
under certain assumptions, identification does not depend on subsets of unit types with constant response.
Finally, we derive rules describing how unit classifications change as assignment rules become more complex.
\subsection{Regression Discontinuity Design}
RDD builds on the groundwork of \cite{Hahn2001}, which provides a formal justification for effect identification and estimation.
An RDD consists of three main elements:
a score $X$ that measures a characteristic of the units, a cutoff $c$ that splits the support of the score into two groups,
and a treatment $D$ that is assigned to each unit depending on whether its score lies above or below the cutoff, $D_i = \ind[X_i\ge c]$ \citep{Cattaneo2019}. \\

Typically, the parameter of interest is the average treatment effect at the cutoff $c$,
$\mathbb{E}[Y(1)-Y(0)\,|\,X=c]$, where $Y(1)$ and $Y(0)$ denote the potential outcomes under treatment and control, respectively \cite{rubin2005}.
Identification of this effect relies on the following key assumption:
\begin{assumption}[Continuity]\label{assumption:std_rdd_continuity}
The conditional mean of the potential outcomes, $\mathbb{E}[Y_i(d)\mid X_i = x]$ for $d\in \{0,1\}$, is continuous at the cutoff $c$.
\end{assumption}

This assumption implies that units with scores close to but on opposite sides of the cutoff are comparable in all relevant aspects \citep{Hahn2001}. 
To rule out selection on gains, it is further required that units near the threshold cannot perfectly manipulate their scores.
Formally, for sufficiently small \(\epsilon > 0\), $\mathbb{E}[Y(1)-Y(0)\mid D,\, X=x] = \mathbb{E}[Y(1)-Y(0)\mid X=x]$ must hold for $x\in (c-\varepsilon, c+\varepsilon)$ \citep{chernozhukov2024}.

Under these relatively mild assumptions, RDD provides inference around the threshold that is as credible as inference from a randomized experiment \citep{Lee2008}.
The average treatment effect at the cutoff,
\begin{equation}
 \tau_0 = \mathbb{E}[Y_i(1)-Y_i(0)\mid X_i = c]
\end{equation}
can be identified as
\begin{equation}
\tau_0 = \lim_{x \to c^+} \mathbb{E}[Y_i \mid X_i = x] - \lim_{x \to c^-} \mathbb{E}[Y_i \mid X_i = x]
\end{equation}
\cite{Hahn2001}.
\subsection{Problem description}
While standard RDD is suitable for analyzing the effect of a single score variable at a cutoff,
real-world treatment decisions often depend on multiple score variables.
For example, treatment \(T_i = 1\) may be assigned to unit \(i\) only if both criteria \(X_{1,i} > c_1\) and \(X_{2,i} > c_2\) are met.
Such multi-score cases are often reduced to a standard RDD by transforming the variables into a single composite score,
allowing identification and estimation to be justified by results from the one-dimensional case.
A drawback of this approach, however, is the reduced interpretability of the transformed score.
This raises the question of the precise conditions under which a treatment effects exist
and can be identified in a multi-score setting.  
In particular, a well-defined multidimensional cutoff effect
requires some independence of the directions from which the cutoff is approached.

In applications, the aim is usually to optimize the treatment assignment \(T_i\) in order to improve the outcome \(Y_i\).
In some critical systems, however, only gradual adjustments to \(T_i\) are feasible.
In such cases, part of the assignment rule \(T_i\) may be treated as fixed;
for example \(H_i \defeq \ind[X_{2,i} > c_2]\) is accepted as given,  
while the remaining part, that is \(G_i \defeq \ind[X_{1,i} > c_1]\), is subject to optimization. 
The effect of a subrule \(G\) at its cutoff can therefore provide insights
into how the overall assignment rule \(T\) might be improved.
In some settings, the assignment rule \(T\) serves only as a recommendation, 
and the actual treatment \(D\) may deviate from it.
In the RDD literature, this case is referred to as ``fuzzy''.
Altogether, this leads to a hierarchy of decision rules \(G \to T \to D\). 
Each step in the hierarchy obscures the assessment of \(G\) (for example, evaluation of the optimal value of cutoff \(c_1\))
because at every level of it a unit may become non-compliant.
Knowledge of the unit's compliance type with respect to \((G,\, D)\), however allows these ambiguities to be addressed.
In practice, this is usually handled using a fuzzy RDD estimator, which adjusts 
the sharp effect estimate by incorporating an estimate of the treatment probability \(\Pr(D_i=1 \,|\, G_i)\), 
thereby recovering a complier effect.

In contrast, we categorize unit behavior with respect to the rules \((G,\, D)\)
into the standard compliance types: compliers, defiers, nevertakers, alwaystakers,
and an additional category of indecisive units.
Based on this categorization, we derive assumptions for identifying the 
complier effect with respect to \((G,\,D)\).
We then examine conditions under which non-change units (e.g., nevertakers) can be excluded
without affecting identification and thus avoiding some of the ambiguities.
For example, units with \(X_{2,i} \leq c_2\) always have \(T_i = 0\) regardless of \(G_i\) and thus  
may not contribute to the complier effect of \(G\).
In some cases, this subsetting principle justifies the use of sharp RDD estimators,
which are typically more stable than their fuzzy counterparts.
Finally, we show how the unit classifications change when moving
on the hierarchy from simple to more complex assignment rules,
incorporating knowledge of intermediate categorizations.

\subsection{Setup}
We use the notation
\(\Rn^K_{> 0} \defeq \left\{x \in \Rn^K \,\middle|\, \forall \, 0 \leq j \leq K \,:\, x_j > 0 \right\}\)
for vectors with positive components
and adapt it in the canonical way for the cross product of
other ordered sets and relations (e.g., \(\Rn^K_{< y}\) or \(\Nn_{\leq k}\)).
Let \(e_1, \ldots, e_K \in \Rn^K\) denote the standard basis vectors,
and let \(\sprodk{\cdot}{\cdot}\) the standard inner product in \(\Rn^K\).
For linear spaces \(U\) and \(V\) over \(\Rn\) with \(U \cap V = \{0\}\), we write \(U \oplus V\)
for their direct sum.

For each individual \(i\), let \(X_i = (X_{1,i}, \, \ldots, \, X_{K,i}) \in \Rn^K\) 
denote the vector of score variables. 
Given \(c \in \Rn^K\), let
\(I_i(c) \defeq (I_{1,i}(c_1), \, \ldots, \, I_{K,i}(c_K))\)
define the indicator vector, where
\(I_{k,i}(c_k) \defeq \ind[X_{k,i} > c_k]\).
Let the observed outcome for individual \(i\) be \(Y_i\).
We treat the entries of \(I_i(c)\) as Boolean variables and allow 
any composition of AND, OR, and negation (\(\land, \, \lor, \overline{\square}\))
over this set of atoms to form general Boolean functions \(g(I_i(c))\).
\begin{definition}
  A mapping \(T: \Rn^K \to \{0, 1\}\) is called a decision rule.
  We say that a decision rule \(T\) is a cutoff rule on \(\Rn^K\) if there exists a Boolean mapping \(g\) such that
  \(T_i = T(X_i) = g(I_i)\).
  That is, \(T\) determines the individual decisions \(T_i\) by applying cutoffs to the components of \(X_i\).
\end{definition}
With slight abuse of notation, we use \(T(c)\), \(T(X_i \,|\, c)\) and \(T_i(c)\) to indicate the use of a specific cutoff \(c \in \Rn^K\).
Note that   
\begin{align}\label{eq:transf_cutoff}
  T(X + \epsilon \,|\, c) = T(X \,|\, c - \epsilon) \quad 
  \mbox{and} \quad
  T(\lambda X \,|\, c) = T\left(X\,\middle|\, \frac{c}{\lambda}\right)
\end{align}
holds for \(\epsilon \in \Rn^K, \lambda \in \Rn_{>0}\).
Thus, without loss of generality, we assume the cutoff of interest is \(c = 0\).
We suppose that the outcome \(Y_i\) depends in the following way on a cutoff rule \(T\) and a general decision rule \(D\):
\begin{align*}
  Y_i = Y_i(T_i,\, D_i) &= \bigg(Y_i(0,\, 0)(1-D_i) + Y_i(0,\, 1)D_i\bigg)(1-T_i)\\
                      &+ \bigg(Y_i(1,\, 0)(1-D_i) + Y_i(1,\, 1)D_i\bigg)T_i,
\end{align*}
where \(T\) is the treatment assignment and \(D\) is the implemented treatment.
Unless otherwise stated, we impose no further assumptions on \(D\), except that it is a decision rule.
\subsection{Unit Categorization}
Given this setup, certain groups of individuals are of particular interest:
nevertakers, alwaystakers, compliers and defiers with respect to the pair \((T,\, D)\) or to \((G,\, D)\), where \(G\) is a subrule of \(T\).
A meaningful categorization of unit \(i\)'s requires counterfactual reasoning:
we must consider how hypothetical changes in the cutoff (or, equivalently, in the observed score values)
would have changed the responses \(T_i\) and \(D_i\).
\begin{figure}[ht]
  \centering
  \begin{tikzpicture}
  \begin{axis}[
    xmax=2,ymax=2,xmin=-2,ymin=-2,
    axis lines=middle,
    axis line style={Stealth-Stealth,thick},
    xtick distance=1,
    ytick distance=1,
    yticklabel=\empty,
    xticklabel=\empty,
    grid=major,
    grid style={thin,densely dotted,black!20},
    ylabel=\small $X_2$,
    xlabel=\small $X_1$
  ]
    \draw [orange] (1,1.2) node {\small $T=1$};
    \draw [orange] (1,-1.2) node {\small $T=1$};

    \draw [orange] (-1,1.2) node {\small $T=0$};
    \draw [orange] (-1,-1.2) node {\small $T=0$};

    \draw [cyan] (-1, 1.6) node {\small $D=0$};
    \draw [cyan] (1,-1.6) node {\small $D=0$};
    \draw [cyan] (-1,-1.6) node {\small $D=0$};
    \draw [cyan] (1,1.6) node {\small $D=1$};

    \node at (-0.9, 0.25)[circle,fill,inner sep=1.6pt] (X) {};
    \node at (0.6, 0.25)[draw,circle,gray!90,inner sep=1.6pt,fill=gray!90] (Xc) {};
    \draw[solid,densely dotted,color=gray!40,very thick,->] (X) edge (Xc); 
    \node[above=0.02cm of X] (XT) {\small \(X_i\)};
    \node[above=0.02cm of Xc,color=gray!90] (XcT) {\small \(X_i-c\)};
  \end{axis}
\end{tikzpicture}
  \caption{
    Cutoff rules \(D = I_1 \land I_2\) and \(T = I_1\).
    The decision boundaries coincide with the coordinate axes.
    When \(X_2 > 0\), \(T\) complies with \(D\);
    otherwise, \(D = 0\) regardless of the value of \(X_1\), which governs the behavior of \(T\).
  }
  \label{fig:cutoff_shift_and}
\end{figure}
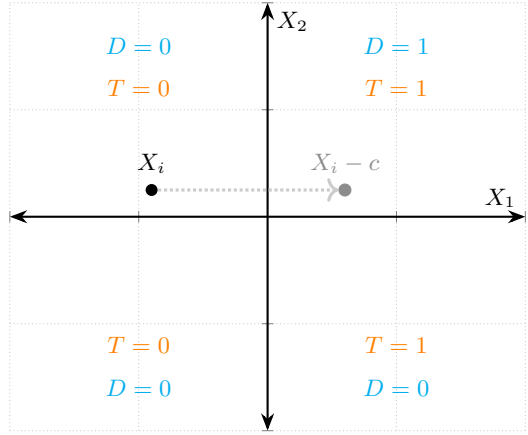
Figure \ref{fig:cutoff_shift_and} illustrates this idea for \(D \defeq I_1 \land I_2\) and \(T \defeq I_1\),
and shows that only changes \(c \in \Rn^2\) that affect \(T\) are relevant for categorizing the behavior of a unit \(i\) with respect to \((T, \, D)\).
Additional variables that affect only \(D\) may not be controllable or even observable.
The following definition formalizes the intuition behind relevant directions of change in a cutoff rule.
\begin{definition}
  Let \(T\) be a cutoff rule on \(\Rn^K\).
  Define
  \[S(T) \defeq \left\{k \in \Nn_{\leq K} \,\middle|\, \exists c \in \Rn^K, \lambda \in \Rn: \, T(0 \,|\, c + \lambda e_k) \neq T(0 \,|\, c)\right\}\]
  as the set of support directions of \(T\).
  The support of \(T\) is then defined as the linear span of these directions:
  \[\supp(T) \defeq \left\{\sum_{k \in S(T)} \lambda_k e_k \,\middle|\, \lambda_k \in \Rn,\, k \in S(T) \right\}\]
\end{definition}
This technical definition greatly simplifies the treatment of equivalent cutoff rules
for example when \(T \defeq I_1 \land (I_2 \lor \overline{I_2})\) and \(G \defeq I_1\).
It also makes it possible to compare \(T\) and \(D\) on their common domain.
The notion of the support of a cutoff rule is motivated by the existence of changes along coordinate directions.
It follows almost immediately that if the support is empty,
then no changes are possible along linear combinations of the coordinate directions.
This property justifies the definition.
\begin{restatable}{lemma}{lemconstsupp} \label{lem:constsupp}
  \(T\) is constant if and only if \(S(T) = \emptyset\), or equivalently, if and only if \(\supp(T) = \{0\}\).
\end{restatable}
From this point on, we assume that \(T\) is non-degenerate, i.e., \(\supp(T) \neq \{0\}\).
In the introductory AND-rule example,
with \(D \defeq I_1(c) \land I_2(c)\) and \(T \defeq I_1(c)\),
we obtain \(\supp(T) = \Rn \times \{ 0 \}\) and \(\supp(D) = \Rn^2\).
If a unit \(i\) complies with the assigned treatment \(T\) relative to the actual treatment \(D\),
then \(T_i\) should coincide with \(D_i\) even under hypothetical changes of the cutoff \(c_1\),
as illustrated in Figure \ref{fig:cutoff_shift_and}.
A more refined, local definition could weaken this requirement to  ``any reasonable changes'' (e.g. perturbations in a neighborhood of the cutoff),
but for clarity we adopt the global one.
In other words, both rules should agree on the support of \(T\), which in this case occurs if \(X_{2,i} > c_2\).
Thus, it is \(X_{2,i}\) and \(c_2\) that govern the behavior of unit \(i\).
To make this distinction explicit, it is useful to introduce notation
for entries \(X \in \Rn^K\) that do not affect \(T\): 
\[N^T \defeq \left\{ X \in \Rn^K \,\middle|\, T(X \,|\, c) = T(0\,|\,c) \, \mbox{for all} \, c \in \Rn^K \right\}\]
In particular, one can show the following.
\begin{restatable}{prop}{propdecomposition}\label{prop:decomposition}
For each \(X \in \Rn^K\), there exists a unique decomposition \(X = X^T + X^{\perp T}\)
with \(X^T \in \supp(T)\) and \(X^{\perp T} \in N^T\).
The decomposition given by the orthogonal projection \(P_T(X) \defeq \sum_{k \in S(T)} \sprodk{X}{e_k}e_k \) onto \(\supp(T)\)
satisfies the above properties.
\end{restatable}
Thus, according to Proposition \ref{prop:decomposition}, the score space decomposes as
\[
  \Rn^{|S(T)|} \times \Rn^{K-|S(T)|}
  \simeq\footnote{Both linear spaces are isomorphic}
  \supp(T) \, \oplus N^T = \Rn^K.
\]
The first component captures the behavior of \(T\),
whereas the second component consists of free variables that do not affect \(T\),
but may influence the unit category.
This observation motivates the following general definition of the unit categories.
\begin{definition}\label{def:unit_categories}
  Let \(i\) be a unit and let \(X_i = X_i^T + X_i^{\perp T}\). 
  Then \(i\) is said to be:
  \begin{enumerate}
    \item a \hldef{nevertaker} (of \(T\) with respect to \(D\)) iff \(D(X_i^{\perp T} - c) = 0\)
    \item an \hldef{alwaystaker} (of \(T\) with respect to \(D\)) iff \(D(X_i^{\perp T} - c) = 1\)
    \item a \hldef{complier} (of \(T\) with respect to \(D\)) iff \(T(0\,|\, c) = D(X_i^{\perp T} - c)\)
    \item a \hldef{defier} (of \(T\) with respect to \(D\)) iff \(T(0\,|\, c) \neq D(X_i^{\perp T} - c)\)
  \end{enumerate}
  for all \(c \in \supp(T)\).
  Let the sets of alwaystakers, nevertakers, compliers and defiers be denoted by
  \(\alwayst(T, D)\), \(\nevert(T, D)\), \(\comp(T, D)\) and \(\defier(T, D)\).
\end{definition}
In the special case where \(D\) is a cutoff rule, these definitions admit the following intuitive equivalencies, 
which capture the notion of simultaneous cutoff changes along the relevant directions.
\begin{restatable}{prop}{proequivcutoffrules}\label{prop:equiv:cutoffrules}
  Let \(D\) be a cutoff rule on \(\Rn^K\).
  The unit \(i\) is:
  \begin{enumerate}
    \item a \hldef{nevertaker} (of \(T\) with respect to \(D\)) iff \(D_i(c) = 0\)
    \item an \hldef{alwaystaker} (of \(T\) with respect to \(D\)) iff \(D_i(c) = 1\)
    \item a \hldef{complier} (of \(T\) with respect to \(D\)) iff \(T_i(c) = D_i(c)\)
    \item a \hldef{defier} (of \(T\) with respect to \(D\)) iff \(D_i(c) \neq T_i(c)\)
  \end{enumerate}
  for all \(c \in \supp(T)\).
\end{restatable}
Definition \ref{def:unit_categories} indeed introduces well-defined categories, which can be summarized as follows:
\begin{restatable}{prop}{propdisjoint}\label{prop:disjoint}
  The sets \(\alwayst(T,\, D),\, \nevert(T,\, D),\, \comp(T,\, D)\) and \(\defier(T,\, D)\)
  are  pairwise disjoint.
\end{restatable}
In general, these categories are not exhaustive,
since \(D\) may not be constant nor equal to \(T\) or \(\overline{T}\) on the support of \(T\).
We call the remaining category 
\hldef{indecisive} and denote
the set of \hldef{indecisive units (of \(T\) with respect to \(D\))} by \(\indes(T,\, D)\).
A unit \(i\) is indecisive iff
\(c \mapsto D(X_i^{\perp T}-c)\) is not constant on \(\supp(T)\)
and
there exist \(c,\, \hat{c} \in \supp(T)\) such that
\(T(X_i \,|\, c) = D(X_i^{\perp T} - c) \)
and
\(T(X_i \,|\, \hat{c}) \neq D(X_i^{\perp T} - \hat{c})\).
Figure \ref{fig:cutoff_general} visualizes the general case in which \(D\) is not a cutoff rule.
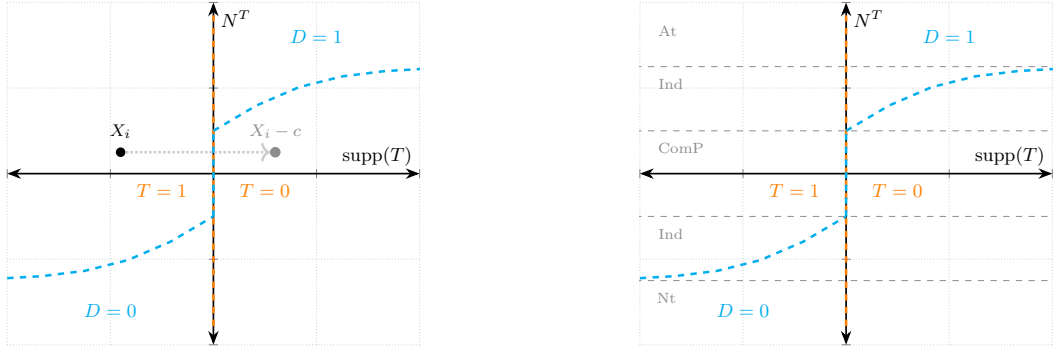
\begin{figure}[ht!]
  \centering
  \begin{subfigure}[ht]{0.5\linewidth}
    \begin{adjustbox}{minipage=\linewidth, scale=0.8}
      \begin{tikzpicture}
  \begin{axis}[
    xmax=2,ymax=2,xmin=-2,ymin=-2,
    axis lines=middle,
    axis line style={Stealth-Stealth,thick},
    xtick distance=1,
    ytick distance=1,
    yticklabel=\empty,
    xticklabel=\empty,
    grid=major,
    grid style={thin,densely dotted,black!20},
    ylabel=\small $N^T$,
    xlabel=\small $\supp(T)$,
    declare function = {
      border_d(\t) = (\t < 0) * (0.75 * tanh(\t) - 0.5) + (\t >= 0) * (0.75 * tanh(\t) + 0.5);
    }
  ]
    \draw [orange,dashed,very thick] plot coordinates {(0,1.85)(0,-1.85)};
    \draw [orange] (0.5,-0.2) node {\small $T=0$};
    \draw [orange] (-0.5,-0.2) node {\small $T=1$};

    \addplot[dashed,draw=cyan,mark=none,very thick,domain=-10:0]{0.75*tanh(\x)-0.5};
    \addplot[dashed,draw=cyan,mark=none,very thick,domain=0:10]{0.75*tanh(\x)+0.5};
    \draw[dashed,draw=cyan,mark=non,very thick] (0, -0.5) -- (0, 0.5);
    \draw [cyan] (-1,-1.60) node {\small $D=0$};
    \draw [cyan] (1,1.60) node {\small $D=1$};

    \node at (-0.9, 0.25)[circle,fill,inner sep=1.6pt] (X) {};
    \node at (0.6, 0.25)[draw,circle,gray!90,inner sep=1.6pt,fill=gray!90] (Xc) {};
    \draw[solid,densely dotted,color=gray!40,very thick,->] (X) edge (Xc); 
    \node[above=0.02cm of X] (XT) {\footnotesize \(X_i\)};
    \node[above=0.02cm of Xc, color=gray!90] (XcT) {\footnotesize \(X_i-c\)};



    %
  \end{axis}
\end{tikzpicture}
    \end{adjustbox}
    \centering
    \caption{
      No matter which change \(c\) of \(X_i\) in the \(\supp(T)\) plane we imagine,
      the responses of \(D\) and \(T\) are equal. 
      A shift of score \(X_i\) by \(c\) is equivalent to a simultaneous shift 
      of the decision boundaries by \(c\), which amounts to a coordinate shift for \(D\) and a 
      cutoff shift for \(T\) motivating Definition \ref{def:unit_categories}.
    }
  \end{subfigure}%
  ~
  \begin{subfigure}[ht]{0.5\linewidth}
    \begin{adjustbox}{minipage=\linewidth, scale=0.8}
      \begin{tikzpicture}
  \begin{axis}[
    xmax=2,ymax=2,xmin=-2,ymin=-2,
    axis lines=middle,
    axis line style={Stealth-Stealth,thick},
    xtick distance=1,
    ytick distance=1,
    yticklabel=\empty,
    xticklabel=\empty,
    grid=major,
    grid style={thin,densely dotted,black!20},
    ylabel=\small $N^T$,
    xlabel=\small $\supp(T)$,
    declare function = {
      border_d(\t) = (\t < 0) * (0.75 * tanh(\t) - 0.5) + (\t >= 0) * (0.75 * tanh(\t) + 0.5);
    }
  ]
    \draw [orange,dashed,very thick] plot coordinates {(0,1.85)(0,-1.85)};
    \draw [orange] (0.5,-0.2) node {\small $T=0$};
    \draw [orange] (-0.5,-0.2) node {\small $T=1$};

    \addplot[dashed,draw=cyan,mark=none,very thick,domain=-10:0]{0.75*tanh(\x)-0.5};
    \addplot[dashed,draw=cyan,mark=none,very thick,domain=0:10]{0.75*tanh(\x)+0.5};
    \draw[dashed,draw=cyan,mark=non,very thick] (0, -0.5) -- (0, 0.5);
    \draw [cyan] (-1,-1.60) node {\small $D=0$};
    \draw [cyan] (1,1.60) node {\small $D=1$};

    \addplot[dashed,draw=gray!90,mark=none]{0.5};
    \addplot[dashed,draw=gray!90,mark=none]{-0.5};
    \addplot[dashed,draw=gray!90,mark=none]{1.25};
    \addplot[dashed,draw=gray!90,mark=none]{-1.25};

    \draw [text=gray!90] (-1.73,1.67) node {\scriptsize $\alwayst$};
    \draw [text=gray!90] (-1.7,1.05) node {\scriptsize $\indes$};
    \draw [text=gray!90] (-1.6,0.3) node {\scriptsize $\comp$};
    \draw [text=gray!90] (-1.7,-0.7) node {\scriptsize $\indes$};
    \draw [text=gray!90] (-1.74,-1.45) node {\scriptsize $\nevert$};



    %
  \end{axis}
\end{tikzpicture}
    \end{adjustbox}
    \centering
    \caption{
      Unit categorizations with respect to \((T, \, D)\).
      While \(\supp(T)\) captures the directions in which \(T\) can change,
      unit behavior is determined by \(N^T\).
      The space is partitioned (from top to bottom) into areas
      of alwaystakers, indecisives, compliers, indecisives and nevertakers of \(T\) with respect to \(D\).
    }
  \end{subfigure}
  \caption{
    General case with \(D\) not being a cutoff rule.
    The decision boundaries of \(T\) and \(D\) are indicated with orange and blue dashed lines.
  }
  \label{fig:cutoff_general}
\end{figure}

From this point forward, whenever we assume that \(D\) is a cutoff rule,
we restrict attention to the cases in which \(D\) contains at least as much information
for decision-making as \(T\).
This means that \(D\) might depend on \(I_{k,i}(0)\) for \(k \in S(T)\). 
In addition, we suppose that \(D\) does not depend on \(I_{k,i}(c)\) with \(c \neq 0\)
for \(k \in \supp(T)\), effectively restricting attention to rules with zero cutoff.
This restriction narrows the general case.
For example, suppose the decision-maker \(D\) behaves opportunistically, ignoring all but one
score \(X_k\), only complying with \(T\) when \(X_k\) exceeds some higher cutoff \(0 < c_k < X_k\).
%

Even for cutoff rules \(D\) under these restrictions, indecisive units are possible.
For example, let \(D \defeq (\overline{I}_1 \land I_2)\) and \(T \defeq I_1 \land I_2\).
Then \(\supp(T) = \Rn^2 = \supp(D)\) and for \(T(0 \,|\, 0) = D(0 \,|\, 0)\) but \(0 = T(0 \,| (-1, 1)) \neq D(0 \,|\, (-1, 1)) = 1\).
Thus, \(D\) is not constant and differs from \(T\) and \(\overline{T}\) on the support of \(T\).
Given \(T\) with \(\dim(\supp(T)) > 1\), one can always construct a cutoff rule \(D\) that produces indecisive items.
At least one can prove the following:
\begin{restatable}{prop}{propexhaustive}\label{prop:exhaustive}
  Let \(i\) denote an individual.
  Then \(i \in \alwayst(T, D) \cup \nevert(T, D) \cup \defier(T, D) \cup \comp(T, D)\) for all cutoff rules \(D\)
  if and only if \(\dim(\supp(T)) = 1\).
\end{restatable}
Moreover, our definitions of alwaystakers, nevertakers and compliers imply
the corresponding definitions in \cite{Imbens2008}.
For this, let \(c^+, c^- \in \supp(T)\) be two directions that induce change in \(T\), so that:
\[
  \lim_{\lambda \to 0} T(X_i \,|\, \lambda c^+) = 1
  \, \mbox{ and } \,
  \lim_{\lambda \to 0} T(X_i \,|\, \lambda c^-) = 0
\]
Using the above complier definition
\[
  T(X_i \,|\, c)
  = T(X_i^T \,|\, c) = T(0 \,|\, c - X_i^T)
  = D(X_i^{\perp T} + X_i^T - c) = D(X_i - c)
\]
for \(c \in \supp(T)\). 
In particular, \(T\) and \(D\) coincide in a neighborhood of zero,
and thus
\[ 
  \lim_{\lambda \to 0} D_i(\lambda c^+) = \lim_{\lambda \to 0} D(X_i - \lambda c^+) = 1
  \, \mbox{ and } \,
  \lim_{\lambda \to 0} D_i(\lambda c^-) = \lim_{\lambda \to 0} D(X_i - \lambda c^-) = 0
\]
where \(D_i(c) \defeq D(X_i - c)\) for \(c \in \supp(T)\).
The corresponding statements for nevertakers and alwaystakers follow analogously.
By requiring consistency in the limit for any direction, one arrives at a local
 definition of unit categories that is sufficient for the multi-score RDD setting.
Before turning to the identification results, we expand on the AND-rule example.

\paragraph{Example (AND-Rules)}\label{ex:and_rule}
Let \(D \defeq \bigwedge_{j=1}^K I_{j}\) and \(T \defeq \bigwedge_{j=1}^k I_{j} \) for \(k \in \Nn_{< K}\).
Then \(\supp(T) = \Rn^k \times \{0\}^{K-k}\), 
and the potential outcomes reduce to 
\[Y_i  = Y_i(0, 0)(1-T_i) + Y_i(1, 0)(1-D_i)T_i + Y(1, 1)D_i\],
dropping one of the mixed terms. Additionally, the unit categorizations are as follows:
\begin{enumerate}
  \item \(\comp(T,\, D) = \{i \,|\, X_i \in \Rn^k \times \Rn^{K-k}_{>0} \}\)

    Note that \(X_i \in \Rn^k \times \Rn^{K-k}_{>0}\) is equivalent to \(\land_{j=k+1}^K I_{j,i}(c)\) being one
    for \(c \in \supp(T)\). Thus, the rule \(D_i\) effectively reduces to \(T_i\).

  \item \(\alwayst(T,\, D) = \emptyset\)

    Since for each \(X_i \in \Rn^K\) there exists a \(c \in \supp(T)\) such that \(D_i(c) = 0\).
    For example choose any \(c\) with \(X_{1,i} \leq c_1\) and \(c_j = 0\) for \(j \neq 1\).

  \item \(\nevert(T,\, D) = \{i \,|\, \exists  k < j \leq K: \, X_{j,i} \leq 0 \} = \{i \,|\, X_i \in \Rn^k \times (\Rn^{K-k} \setminus \Rn^{K-k}_{>0}) \} \)

    A unit \(i\) being in the set on the right side is equivalent to \(\land_{j=k+1}^K I_{j,i}(c)\) being zero
    for all \(c \in \supp(T)\). This is equivalent to \(D_i(c) = 0\) for all \(c \in \supp(T)\),
    since for this unit one can always find a cutoff \(c \in \supp(T)\) such that \(T_i(c) = 1\), 
    which implies that one of the indicators \(I_{j,i}\) for \(j > k\) has to be zero.

  \item \(\defier(T, \, D) = \emptyset\)

    Note that \(T_i(c) = 0\) implies \(D_i(c) = 0\) for \(c \in \supp(T)\).
    Now suppose that \(T_i(c) = 1\) and \(D_i(c) = 0\) for some \(c \in \supp(T)\).
    Thus, \(I_{j,i}(c) = 0\) for some \(j > k\).
    Since \(I_{j,i}\) independent of \(c \in \supp(T)\), \(D_i = 0\) on \(\supp(T)\).
\end{enumerate}
Further examples can be found in Appendix \ref{appendix:examples},
which introduces additional instances of the remaining unit categories (excluding the indecisive case).
We conclude that the definitions introduced above are reasonable.
\subsection{Effect Identification}
Inspired by the work of \cite{Hahn2001} and \cite{Imbens1994},
we use the unit categories introduced above to prove an identification theorem for the local complier effect at the cutoff.
Throughout this section, we require that the outcome \(Y\) does not directly depend on the treatment assignment \(T\).
Let the set of all unit categories be
\[
  \mathcal{C} \defeq \{
    \comp(T,\,D),\, \nevert(T,\,D),\, \alwayst(T,\,D),\, \defier(T,\,D),\, \indes(T,\,D)
  \}
\]
and define the set of non-change categories as
\[\mathcal{C}^0 \defeq \{\nevert(T,\,D), \, \alwayst(T,\,D) \}.\]
We assume that the categorization of a unit is independent of the support part of \(T\)
in a neighborhood of the cutoff, that is:
\begin{restatable}{assumption}{assumptionindependence}\label{assumption:independence}
  There exists an \(\epsilon > 0\) such that
  \(
    \Pr\left(i \in \cat \,\middle|\, X_i^T = x\right) =
    \Pr\left(i \in \cat \,\middle|\, X_i^T=0\right)
  \)
  for \(\|x\| \leq \epsilon\) and \(\cat \in \mathcal{C}\).
\end{restatable}
This assumption relates to the independence assumptions in \cite{Imbens1994}.
Using Assumption \ref{assumption:independence} and further assuming 
\(\Pr\left(i \in \cat \,\middle|\, X^T_i = 0\right) > 0\)
we obtain
\[
  \mathbb{E}\left(Y_i \,\middle|\, X_i^T = x\right) =
  \sum_{\cat \in \mathcal{C}} \mathbb{E}\left(Y_i \,\middle|\, X_i^T = x, \, i \in \cat\right)
  \Pr\left(i \in \cat \,\middle|\, X_i^T = 0 \right) 
\]
and thus
\begin{align*}
  \mathbb{E}\left(Y_i \,\middle|\, X_i^T = x^+\right) - \mathbb{E}\left(Y_i \,\middle|\, X_i^T  = x^-\right) &=\\
  \sum_{\cat \in \mathcal{C}}\bigg(\mathbb{E}\left(Y_i\,\middle| X_i^T = x^+,\, i \in \cat \right) & - \mathbb{E}\left(Y_i\,\middle| X_i^T = x^-,\, i \in \cat\right) \bigg)
  \Pr\left(i \in \cat \,\middle|\, X_i^T = 0\right)
\end{align*}
We also rely on the following local continuity assumption, which adapts the 
standard continuity condition (Assumption \ref{assumption:std_rdd_continuity}) to the unit categories:
\begin{restatable}{assumption}{assumptionconti}\label{assumption:conti}
  There exists an \(\epsilon > 0\) such that
  \(x \mapsto \mathbb{E}(Y_i(d) \,|\, X_i^T = x, \, i \in \cat)\)
  is continuous for \(\|x\| \leq \epsilon\), \(d \in \{0, 1\}\) and \(\cat \in \mathcal{C}\).
\end{restatable}
This assumption can be weakened by requiring only directional continuity,
in which case the effect would depend on the chosen directions.
Note that
\[ 
  E\left(Y_i \,\middle|\, X_i^T = x^{\pm},\, i \in \alwayst\right) = 
  E\left(Y_i(1) \,\middle|\, X_i^T = x^{\pm}, \, i \in \alwayst\right)
\]
and
\[ 
  E\left(Y_i \,\middle|\, X_i^T = x^{\pm},\, i \in \nevert\right) =
  E\left(Y_i(0) \,\middle|\, X_i^T = x^{\pm}, \, i \in \nevert\right)
\]
holds for the non-change unit categories.
Together with Assumption \ref{assumption:conti} this yields:
\begin{equation}\label{eq:identification:limit}
\begin{aligned}
  \lim_{\lambda \to 0} \mathbb{E}\left(Y_i \,\middle|\, X_i^T = \lambda x^+\right) &
  - \lim_{\lambda \to 0} \mathbb{E}\left(Y_i \,\middle|\, X_i^T  = \lambda x^-\right) =\\
  \sum_{\cat \in \mathcal{C}\setminus \mathcal{C}^0}\bigg(
    \lim_{\lambda \to 0} \mathbb{E}\left(Y_i\,\middle| X_i^T = \lambda x^+,\, i \in \cat \right) & -
    \lim_{\lambda \to 0} \mathbb{E}\left(Y_i\,\middle| X_i^T = \lambda x^-,\, i \in \cat\right)
  \bigg) \cdot \\
  & \cdot \Pr\left(i \in \cat \,\middle|\, X_i^T = 0\right)
\end{aligned}
\end{equation}
Two further assumptions are needed to make use of the continuity condition for the remaining categories.
First, we rule out the existence of indecisive units,
since this category does not allow structured conclusions about \(D\) based on knowledge of \(T\).
In other words, this category does not allow separating the potential outcomes \(Y_i(0)\) and \(Y_i(1)\).
\begin{restatable}{assumption}{assumptionnoindecisives}\label{assumption:no_indecisives}
  \(\indes(T, \, D) = \emptyset\)
\end{restatable}
Second, we assume that the directions \(x^+, \, x^- \in \supp(T)\) 
along which we estimate the complier effect induce a change in \(T\).
\begin{restatable}{assumption}{assumptionchangedirection}\label{assumption:change_direction}
  \(
    1 = \lim_{\lambda \to 0} T\left(\lambda x^+ \,\middle|\, 0\right) 
    \neq  
    \lim_{\lambda \to 0} T\left(\lambda x^- \,\middle|\, 0\right) = 0
  \)
\end{restatable}
This assumption is implicit in one-dimensional RDD designs
and imposes no substantive restriction in practice, as \(T\) is typically known.
With this in place, we know how \(D_i\) behaves for compliers and defiers when approaching
from the \(x^+\) and \(x^-\) directions. That is:
\[
  \lim_{\lambda \to 0}  \Pr\left(D_i=1 \,|\, X_i^T = \lambda x^+,\, i \in \comp\right) = 1
  \,\mbox{ and }\,
  \lim_{\lambda \to 0}  \Pr\left(D_i=1 \,|\, X_i^T = \lambda x^-,\, i \in \comp\right) = 0,
\]
as well as
\[
  \lim_{\lambda \to 0}  \Pr\left(D_i=1 \,|\, X_i^T = \lambda x^+,\, i \in \defier\right) = 0
  \,\mbox{ and }\,
  \lim_{\lambda \to 0}  \Pr\left(D_i=1 \,|\, X_i^T = \lambda x^-,\, i \in \defier\right) = 1.
\]
Thus, we can apply Assumption \ref{assumption:conti}
to these two remaining categories on the right-hand side of Equation \ref{eq:identification:limit} as well:
\begin{theorem}\label{thm:identification}
Let Assumptions \ref{assumption:independence}, \ref{assumption:conti}, 
\ref{assumption:no_indecisives} and \ref{assumption:change_direction} hold.
Then the complier effect at the cutoff is identified as
\begin{equation}\label{eq:identification_main}
  \begin{aligned}
    \mathbb{E}\left(Y_i(1) \,\middle|\, X_i^T = 0,\, i \in \comp\right) 
    &- \mathbb{E}\left(Y_i(0) \,\middle|\, X_i^T = 0,\, i \in \comp\right) = \\
    \frac{1}{\Pr\left(i \in \comp \,\middle|\, X_i^T = 0\right)}
    &\left(
     \lim_{\lambda \to 0} \mathbb{E}\left(Y_i \,\middle|\, X_i^T = \lambda x^+ \right)
    - \lim_{\lambda \to 0} \mathbb{E}\left(Y_i \,\middle|\, X_i^T = \lambda x^- \right) 
    \right) - C
  \end{aligned}
\end{equation}
with \(C\) being the correction term for defiers:
\[
  C \defeq
  \frac{
    \Pr\left(i \in \defier \,\middle|\, X_i^T = 0\right)
  }{
    \Pr\left(i \in \comp \,\middle|\, X_i^T = 0\right)
  }
  \bigg(
   \mathbb{E}\left(Y_i(0) \,\middle|\, X_i^T = 0, \, i \in \defier \right)
  - \mathbb{E}\left(Y_i(1) \,\middle|\, X_i^T = 0, \, i \in \defier \right) 
  \bigg).
\]
\end{theorem}
This identification result has two immediate implications.
First, one is free to choose among the directions \(x^+\) and \(x^-\)
satisfying Assumption \ref{assumption:change_direction}.
Second, the proof suggests that dropping subsets \(\Omega \subset \nevert \cup \alwayst\) 
does not affect identification, as long as doing so does not violate
Assumptions \ref{assumption:conti} and \ref{assumption:independence}.

%
%
Following this idea, we now investigate how excluding units in \(\alwayst \cup \nevert\) affects the above identification result.
For ease of presentation, we assume there are no defiers at the cutoff.
As a consequence, the correction term \(C\) in Equation \ref{eq:identification_main} equals zero.
Now let \(\Omega \subset \alwayst \cup \nevert\).
Then \(\Pr\left(i \in \comp, \, i \in \Omega \,\middle|\, X_i^T = 0\right) = 0\) and thus one has
\begin{align}\label{eq:non_change_drop:denominator}
  \Pr\left(i \in \comp \,\middle|\, X_i^T = 0\right)
  = \Pr\left(i \in \comp \,\middle|\, X_i^T = 0, \, i \notin \Omega\right)  \Pr\left(i \notin \Omega \,\middle|\, X_i^T = 0\right)
\end{align}
for the denominator in Equation \ref{eq:identification_main}.
For the numerator, we have:
\begin{align*}
  \mathbb{E}\left(Y_i \,\middle|\, X_i^T = \lambda x^{\pm}, \, i \in \Omega\right)
  &=  \mathbb{E}\left(Y_i(0) \,\middle|\, X_i^T = \lambda x^{\pm}, \, i \in \Omega \cap \nevert \right) \Pr\left(i \in \nevert \,\middle|\, X_i^T = \lambda x^{\pm}, \, i \in \Omega\right)\\
  &+  \mathbb{E}\left(Y_i(1) \,\middle|\, X_i^T = \lambda x^{\pm}, \, i \in \Omega \cap \alwayst \right) \Pr\left(i \in \alwayst \,\middle|\, X_i^T = \lambda x^{\pm}, \, i \in \Omega\right)
\end{align*}
Requiring Assumptions \ref{assumption:conti} and \ref{assumption:independence} to hold when conditioning
on \(\Omega \cap \nevert\) and \(\Omega \cap \alwayst\) (instead of \(\nevert\) and \(\alwayst\)), we obtain:
\begin{align}\label{eq:non_change_drop:zero_effect}
  \lim_{\lambda \to 0} \mathbb{E}\left(Y_i \,\middle|\, X_i^T = \lambda x^+, \, i \in \Omega\right)
  &- \lim_{\lambda \to 0} \mathbb{E}\left(Y_i \,\middle|\, X_i^T = \lambda x^-, i \in \Omega\right) = 0.
\end{align}
Since 
\begin{align*}
  \mathbb{E}\left(Y_i \,\middle|\, X_i^T = \lambda x^{\pm}\right)
  &= \mathbb{E}\left(Y_i \,\middle|\, X_i^T = \lambda x^{\pm}, \, i \in \Omega\right) \Pr\left(i \in \Omega \,\middle|\, X_i^T = \lambda x^{\pm} \right)\\
  &+ \mathbb{E}\left(Y_i \,\middle|\, X_i^T = \lambda x^{\pm}, \, i \notin \Omega\right)\Pr\left(i \notin \Omega \,\middle|\, X_i^T = \lambda x^{\pm} \right)
\end{align*}
holds, we require an assumption similar to Assumption \ref{assumption:independence} for \(\Pr\left(i \in \Omega \,\middle|\, X_i^T = \lambda x^{\pm} \right)\), 
in order to obtain
\begin{equation}\label{eq:non_change_drop:nominator}
\begin{aligned}
  \lim_{\lambda \to 0} \mathbb{E}\left(Y_i \,\middle|\, X_i^T = \lambda x^+\right)
  &- \lim_{\lambda \to 0} \mathbb{E}\left(Y_i \,\middle|\, X_i^T = \lambda x^-\right) \\
  &=
  \left(
    \lim_{\lambda \to 0} \mathbb{E}\left(Y_i \,\middle|\, X_i^T = \lambda x^+, \, i \notin \Omega\right)
    - \lim_{\lambda \to 0} \mathbb{E}\left(Y_i \,\middle|\, X_i^T = \lambda x^-, \, i \notin \Omega\right)
  \right) \cdot \\
  & \cdot \Pr\left(i \notin \Omega \,\middle|\, X_i^T = 0 \right)
\end{aligned}
\end{equation}
using Equation \ref{eq:non_change_drop:zero_effect}.
Combining Equations \ref{eq:non_change_drop:nominator} and \ref{eq:non_change_drop:denominator}, we obtain the following result:
\begin{theorem}\label{thm:identification:none_change_drop}
  Let Assumptions \ref{assumption:independence}, \ref{assumption:conti}, 
  \ref{assumption:no_indecisives} and \ref{assumption:change_direction} hold,
  and suppose \(\Pr(i \in \defier \,|\, X_i^T = 0) = 0\).
  Further, let \(\Omega \subset \nevert \cup \alwayst\) such that:
  \begin{enumerate}
    \item there exists an \(\epsilon > 0\) such that the functions
      \[x \mapsto \mathbb{E}\left(Y_i(0) \,\middle|\, X_i^T = x, \, i \in \Omega \cap \nevert \right)\]
      and
      \[x \mapsto \mathbb{E}\left(Y_i(1) \,\middle|\, X_i^T = x, \, i \in \Omega \cap \alwayst \right)\]
      are continuous for \(\|x\| < \epsilon\).
    \item there exists an \(\epsilon > 0\) such that
      \[\Pr\left(i \in \nevert \,\middle|\, X_i^T = x, \, i \in \Omega\right) = \Pr\left(i \in \nevert \,\middle|\, X_i^T = 0, \, i \in \Omega\right) \]
      and
      \[\Pr\left(i \in \alwayst \,\middle|\, X_i^T = x, \, i \in \Omega\right) = \Pr\left(i \in \alwayst \,\middle|\, X_i^T = 0, \, i \in \Omega\right) \]
      as well as
      \[\Pr\left(i \in \Omega \,\middle|\, X_i^T = x\right) = \Pr\left(i \in \Omega \,\middle|\, X_i^T = 0 \right) \]
      for \(\|x\| < \epsilon\).
  \end{enumerate}
  Then
  \begin{align*}
    \mathbb{E}\left(Y_i(1) \,\middle|\, X_i^T = 0,\, i \in \comp\right) 
    &- \mathbb{E}\left(Y_i(0) \,\middle|\, X_i^T = 0,\, i \in \comp\right) = \\
    \frac{1}{\Pr\left(i \in \comp \,\middle|\, X_i^T = 0, \, i \notin \Omega\right)}
    &\left(
     \lim_{\lambda \to 0} \mathbb{E}\left(Y_i \,\middle|\, X_i^T = \lambda x^+, \, i \notin \Omega \right)
    - \lim_{\lambda \to 0} \mathbb{E}\left(Y_i \,\middle|\, X_i^T = \lambda x^-, i \notin \Omega \right) 
    \right)
  \end{align*}
  holds.
\end{theorem}
We refer to estimates of the complier effect of \(T\) obtained when removing \(\Omega\)
as the \hldef{subset complier effect of \(T\) (excluding \(\Omega\)}).
Theorem \ref{thm:identification:none_change_drop} provides
sufficient conditions when both effects (the complier effect as identified in Theorem \ref{thm:identification}
and the subset complier effect) are equal.
\subsection{Inheritance}
In this section, we investigate how the categorization of units changes when
the treatment assignment \(T\) is altered, e.g., by replacing it with a subrule \(G\).
The goal is to identify unit behavior without relying on knowledge of unobservable parts of the decision rules.
In particular, we relate unit behavior under complex rules to that under simpler subrules.
We begin by considering the effect of negation:
\begin{restatable}{prop}{propdualities}\label{prop:dualities}
  The following duality statements hold:
  \begin{enumerate}
    \item \(\alwayst(T,\, D) = \nevert(T,\, \overline{D})\)
    \item \(\alwayst(T,\, D) = \alwayst(\overline{T},\, D)\) and 
          \(\nevert(T,\, D) = \nevert(\overline{T},\, D)\)
    \item \(\indes(T,\, D) = \indes(\overline{T},\, D) = \indes(T,\, \overline{D})\)
    \item \(\comp(T,\, D) = \defier(\overline{T},\, D) = \defier(T,\, \overline{D})\)
  \end{enumerate}
\end{restatable}
Indecisive units remain unchanged under negation of either decision rule.
Non-change units are stable under negation of \(T\) but flip when \(D\) is negated. 
In contrast, compliers and defiers always flip.
Next, we investigate the stability of the non-change categories under sub-rules.
In general, one has the following bounds:
\begin{restatable}{prop}{propboundatnt}\label{prop:bound:atnt}
  Let \(G\) be a cutoff rule on \(\Rn^K\) with \(\supp(G) \subset \supp(T)\). Then:
  \begin{enumerate}
    \item \(\alwayst(T, \, D) \subset \alwayst(G, \, D)\)
    \item \(\nevert(T, \, D) \subset \nevert(G, \, D)\)
  \end{enumerate}
\end{restatable}
To obtain similar results for the other categories, assumptions about the
relation between \(G\) and \(T\) are required.
The following proposition generalizes the AND-rule and OR-rule examples
(Example \ref{ex:and_rule} and Example \ref{ex:or_rule}).
\begin{restatable}{prop}{propoperationsimple}\label{prop:operation_simple}
  Let \(G, \, H\) be cutoff rules on \(\Rn^K\) with \(\supp(G) \neq \{0\}\).
  Suppose further that \(\supp(T) = \supp(G) \oplus \supp(H)\).
  \begin{enumerate}
    \item If \(T = G \land H\), then:
    \begin{enumerate}
      \item \(\comp(G,\, T) = \{ i \,|\, H(X_i^H \,|\, 0) = 1 \}\)
      \item \(\nevert(G,\, T) = \{ i \,|\, H(X_i^H \,|\, 0) = 0 \}\)
      \item Every unit \(i\) is either a complier or a nevertaker of \(G\) with respect to \(T\).
    \end{enumerate}
    \item If \(T = G \lor H\), then:
    \begin{enumerate}
      \item \(\comp(G,\, T) = \{ i \,|\, H(X_i^H \,|\, 0) = 0 \}\)
      \item \(\alwayst(G,\, T) = \{ i \,|\, H(X_i^H \,|\, 0) = 1 \}\)
      \item Every unit \(i\) is either a complier or an alwaystaker of \(G\) with respect to \(T\).
    \end{enumerate}
  \end{enumerate}
\end{restatable}
In this setting, if the subrule \(H\) is known, we can identify the unit categories 
with respect to \((G,\, T)\), even if \(G\) itself is not fully observed.
Using Proposition \ref{prop:operation_simple}, we conclude that introducing new score variables
into cutoff rules does not create any defiers or indecisive units:
by definition, negation is required for these categories (see Definition \ref{def:unit_categories} and Proposition \ref{prop:dualities}).
The previous statement can be extended to the case of a general decision rule \(D\) setting as follows:
\begin{restatable}{prop}{propboundcompdef}\label{prop:bound:comp}
  Let \(G, \, H\) be cutoff rules on \(\Rn^K\) with
  \(\supp(T) = \supp(G) \oplus \supp(H)\) and \(\supp(G) \neq \{0\} \neq \supp(H)\).
  \begin{enumerate}
    \item If \(T = G \land H\) then:
      \begin{enumerate}
        \item\label{prop:bound:comp:and_factor}\(\comp(G,\, T) \cap \comp(T,\, D) \subset \comp(G,\, D)\)
        \item\label{prop:bound:comp:nt_factor}\(\nevert(G,\, T) \cap \comp(T,\, D) \subset \nevert(G,\, D)\)
        \item\label{prop:bound:comp:and_bound} \(\comp(T,\, D) \subset \comp(G,\, D) \cup \nevert(G,\, D) \)
      \end{enumerate}
    \item If \(T = G \lor H\) then:
      \begin{enumerate}
        \item\label{prop:bound:comp:or_factor} \(\comp(G,\, T) \cap \comp(T,\, D) \subset \comp(G,\, D)\)
        \item\label{prop:bound:comp:at_factor}\( \alwayst(G,\, T) \cap \comp(T,\, D) \subset \alwayst(G,\, D)\)
        \item\label{prop:bound:comp:or_bound} \(\comp(T,\, D) \subset \comp(G,\, D) \cup \alwayst(G,\, D) \)
      \end{enumerate}
  \end{enumerate}
\end{restatable}
Thus, compliers of \(T \defeq G \, \square \, H\) with respect to \(D\)
are bounded by the complier and non-change categories of the simpler rule \(G\)
as stated in Proposition \ref{prop:bound:comp} \ref{prop:bound:comp:or_bound}
and \ref{prop:bound:comp} \ref{prop:bound:comp:and_bound}.
However, the set equalities do not hold in general.
For example, let \(D \defeq (I_1 \land I_2) \lor I_3\),  \(T \defeq I_1 \land I_3\) and \(G \defeq I_1\).
Then 
\[ \emptyset = \comp(T,\, D) \neq \comp(G,\, D) = \{ i \,|\, X_{3,i} \leq 0, \, X_{2,i} > 0 \} \]
holds.
Furthermore, Proposition \ref{prop:bound:comp} \ref{prop:bound:comp:or_factor},
\ref{prop:bound:comp} \ref{prop:bound:comp:and_factor} 
and \ref{prop:bound:comp} \ref{prop:bound:comp:at_factor},
\ref{prop:bound:comp} \ref{prop:bound:comp:nt_factor} 
resemble a factorization rule:
loosely speaking, the intermediate rule \(T\) in the
hierarchy \((G,\, T)\) and \((T,\, D)\) factors out,  becoming \((G,\, D)\).
Using Propositions \ref{prop:bound:comp} and \ref{prop:dualities}, we can immediately draw conclusions about defiers:
\begin{restatable}{prop}{propbounddef}\label{prop:bound:def}
  Let \(G, \, H\) be cutoff rules on \(\Rn^K\) with
  \(\supp(T) = \supp(G) \oplus \supp(H)\) and \(\supp(G) \neq \{0\} \neq \supp(H)\).
  \begin{enumerate}
    \item If \(T = G \land H\) then:
      \begin{enumerate}
        \item\label{prop:bound:def:and_factor}\(\comp(G,\, T) \cap \defier(T,\, D) \subset \defier(G,\, D)\)
        \item\label{prop:bound:def:nt_factor}\(\nevert(G,\, T) \cap \defier(T,\, D) \subset \alwayst(G,\, D)\)
        \item\label{prop:bound:def:and_bound} \(\defier(T,\, D) \subset \defier(G,\, D) \cup \alwayst(G,\, D) \)
      \end{enumerate}
    \item If \(T = G \lor H\) then:
      \begin{enumerate}
        \item\label{prop:bound:def:or_factor} \(\comp(G,\, T) \cap \defier(T,\, D) \subset \defier(G,\, D)\)
        \item\label{prop:bound:def:at_factor}\( \alwayst(G,\, T) \cap \defier(T,\, D) \subset \nevert(G,\, D)\)
        \item\label{prop:bound:def:or_bound} \(\defier(T,\, D) \subset \defier(G,\, D) \cup \nevert(G,\, D) \)
      \end{enumerate}
  \end{enumerate}
\end{restatable}
Combining both statements, we see that if a unit is a complier with respect to \((T,\, D)\),
its categorization under \((G,\, T)\) also holds for \((G,\, D)\). 
Figure \ref{fig:inheritance} summarizes the above findings for the AND-case.
Further, according to Proposition \ref{prop:operation_simple}, in the absence of indecisive items, these transitions are exhaustive,
since any unit is either a complier or nevertaker of \(G\) with respect to \(T\).
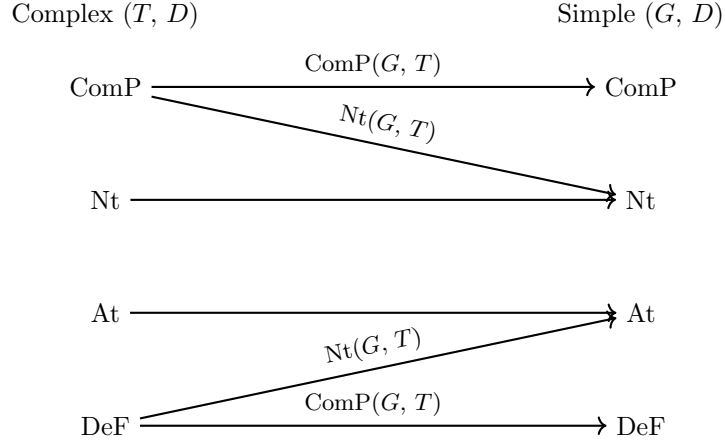
\begin{figure}[ht]
  \centering
  \begin{tikzpicture}
  \begin{scope}[thick]
    \node[] (capComp) {Complex \((T,\,D)\)};
    \node[right=4.5cm of capComp] (capSim) {Simple \((G, \, D)\)};
    \node[below=0.4cm of capComp] (cTD) {\(\comp\)};
    \node[below=of cTD] (nTD) {\(\nevert\)};
    \node[below=of nTD] (aTD) {\(\alwayst\)};
    \node[below=of aTD] (dTD) {\(\defier\)};
    \node[below=0.4cm of capSim] (cGD) {\(\comp\)};
    \node[below=of cGD] (nGD) {\(\nevert\)};
    \node[below=of nGD] (aGD) {\(\alwayst\)};
    \node[below=of aGD] (dGD) {\(\defier\)};
    \draw[->] (nTD) edge node[sloped,above,midway] {} (nGD);
    \draw[->] (aTD) edge node[sloped,above,midway] {} (aGD);

    \draw[->] (dTD) edge node[sloped,above,midway] {\small\(\comp(G, \, T)\)} (dGD);
    \draw[->] (dTD) edge node[sloped,above,midway] {\small\(\nevert(G, \, T)\)} (aGD);

    \draw[->] (cTD) edge node[sloped,above,midway] {\small\(\comp(G, \, T)\)} (cGD);
    \draw[->] (cTD) edge node[sloped,above,midway] {{\small\(\nevert(G, \, T)\)}} (nGD);
  \end{scope}
\end{tikzpicture}




  \caption{
    All possible category changes when transitioning from the complex rule \(T = G \land H\)
    to the simpler subrule \(G\) given that \(\supp(G) \oplus \supp(H) = \supp(T)\) and assuming that there are no indecisive units.
    The figure summarizes Propositions \ref{prop:bound:atnt}, \ref{prop:operation_simple},
    \ref{prop:bound:comp} and \ref{prop:bound:def}.
    The OR-case can be described analogously.
  }
  \label{fig:inheritance}
\end{figure}
%

%

\section{Empirical Application}\label{sec:app}
In this section, we apply our theoretical results to a real-world decision-making scenario in opto-electronic semiconductor manufacturing. We analyze the effect of rework decisions during the color-conversion process in the production of  \abbrev{white light emitting diode}{wLED} on overall product yield. After a brief description of the application, we derive modeling implications based on the observed data. This leads to an algorithmic description of the production process, which serves as a data-generating process (DGP) and complements the empirical study with simulations using artificially generated data that mimic the real-world case.
\subsection{Rework}
Phosphor conversion is a crucial step in wLED production. To obtain white light, several layers of phosphorous substrate are applied to a blue light-emitting semiconductor, shifting the perceived color along a conversion curve from the blue region towards the white region. In this step, each production lot, consisting of $784$ individual wLEDs, is processed according to a standardized recipe. Every wLED in the lot undergoes the same procedure, resulting in an identical stack of phosphor layers.
The goal is to reach the specified target color by the end of production. To achieve this, intermediate color measurements are taken for selected wLEDs in the lot after the conversion step to assess proximity to the target. Subsequent processing steps are then adjusted to maximize the number of wLEDs that meet the target color, provided they are already close. To ensure a high yield (i.e., a large share of wLEDs in the lot reaching the target), a rework decision is made based on these intermediate measurements. If the target is not met, a correction layer of phosphor is applied to the entire lot. For further details on the conversion process, see  \citet{Cho2017} and \citet{schwarz2024}.

The rework decision is based on two scores:
The \textit{distance score}, $X_D$, measures the distance between the mean color point
\(C_1 = (C_x,\, C_y)\) after the regular conversion and the target color
\(C_T := (C_{T, x},\, C_{T, y})\).
The \textit{yield-improvement score}, $X_Y$,
is a relative measure of quality variation within a lot.
It evaluates a hypothetical scenario in which the target is ideally met by the mean color of the lot.
If variability within the lot is high, rework may reduce overall yield even if the distance score suggests treatment.
Figure \ref{fig:scores} illustrates the decision criteria in more detail.

\begin{figure}[htb]
  \centering
  \includegraphics[width=0.8\linewidth]{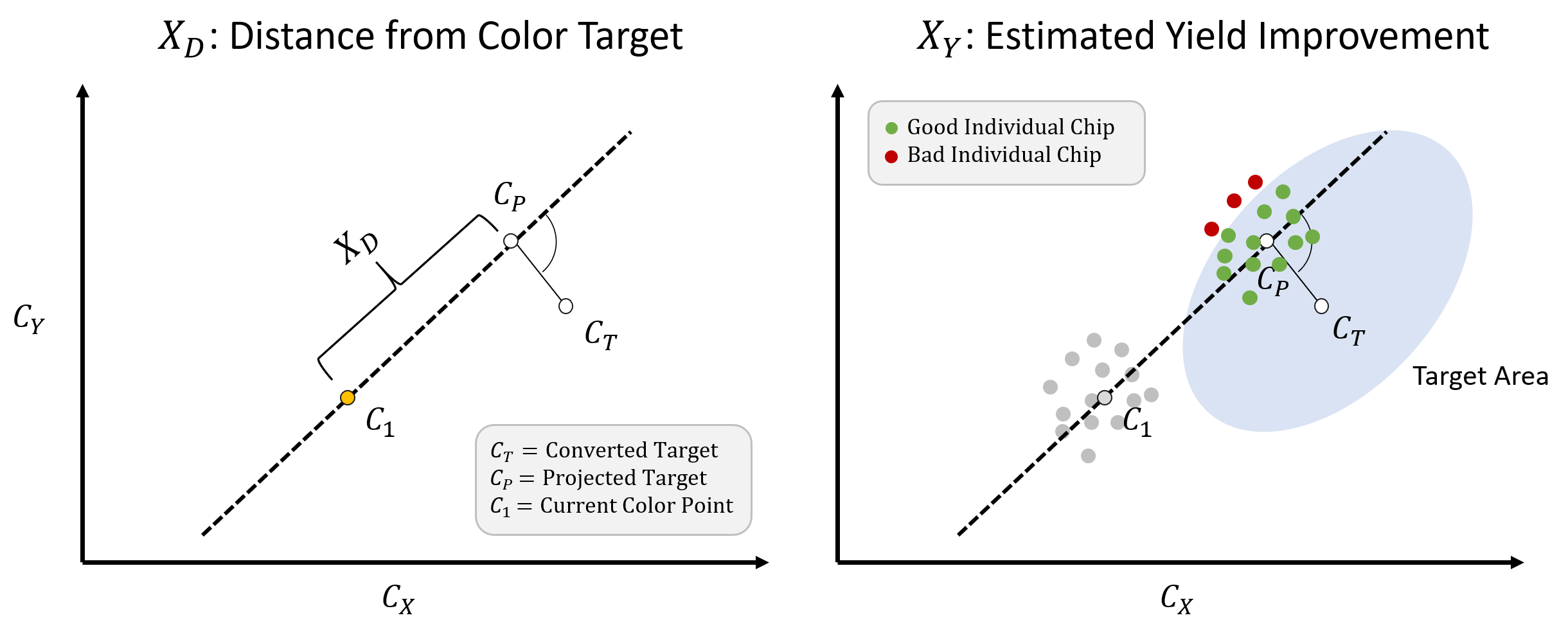}
  \caption{
    Left: $X_D$ is defined as the distance between the current mean color point $C_1$ and the target point $C_P$, which is the closest the lot can technically get to the initial target $C_T$. The slope of the dashed curve is given by a physical process and thus it is not possible to reach $C_T$ once the curve deviates.
    Right: $X_Y$ evaluates the expected improvement by calculating the share of in-specification chips in the lot. This is done by moving the current distribution of color points to the target.
  }
  \label{fig:scores}
\end{figure}

Treatment is assigned according to the cutoff rule
\(T = I_D \land I_Y\), with the goal of maximizing the outcome $Y$,
defined as the percentage of chips that reach the target color by the end of production.

In practice, however, the empirical data show that the operators responsible for performing the treatment $D$
do not always comply with $T$ (see Figure \ref{fig:realdata}).
We attribute this discrepancy to an informational advantage regarding the \(X_Y\) score,
which human operators may use to override $T$ in order to avoid possible yield losses
from applying another rework layer.

We formalize this cautious-operator assumption in the following section.
It reflects the observed one-sided fuzziness in the \(I_Y\) dimension,
as well as the strict compliance with the distance rule \(I_D\).

\begin{figure}[htb]
  \centering
  \includegraphics[width=0.8\linewidth]{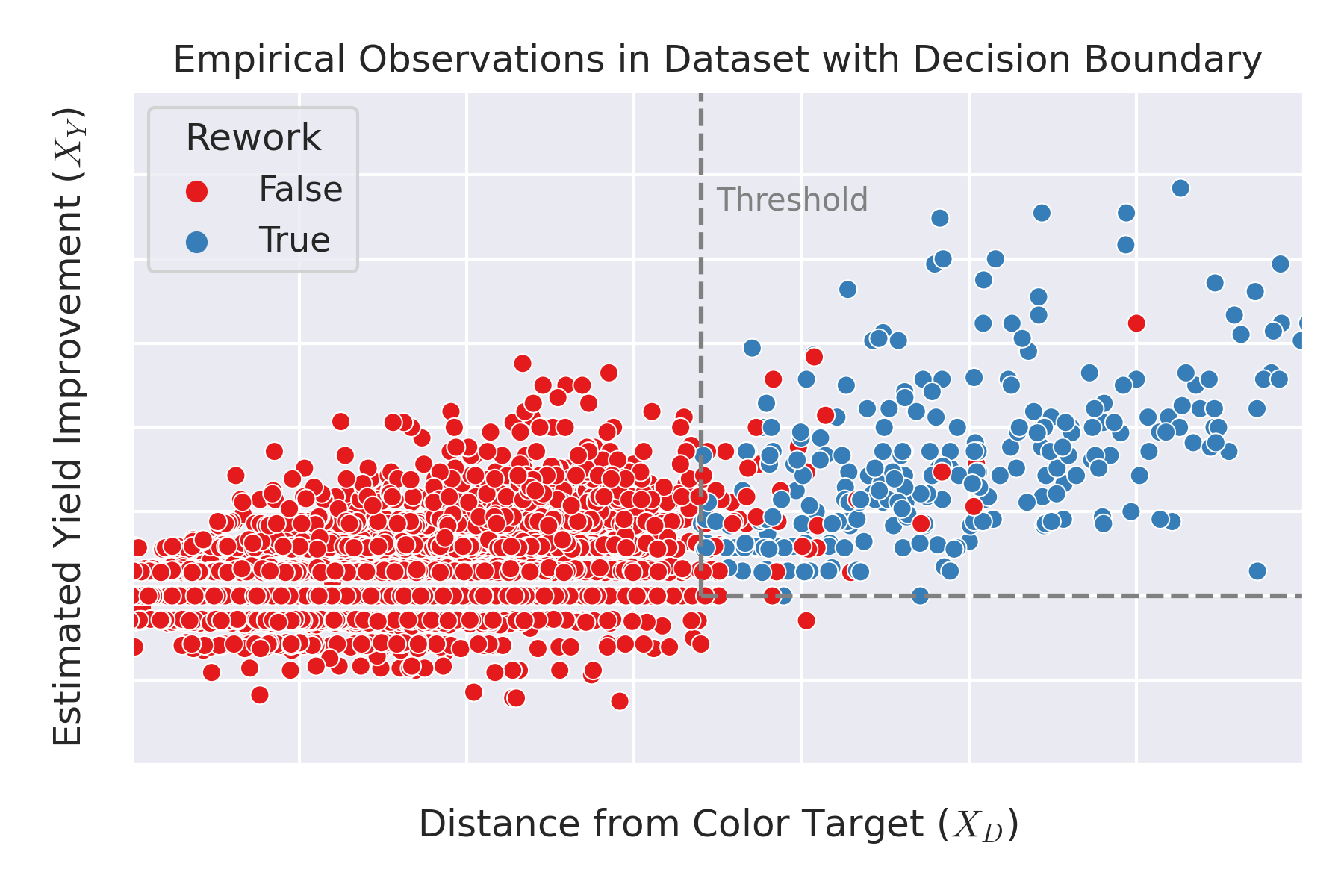}
  \caption{
    Real data plot with respect to the score components $X_D$ and $X_Y$.
    The decision boundary $T$ is dashed.
    The actual treatment assignment (red and blue) follows an unobserved rule $D$,
    rendering the MRD fuzzy.
  }
  \label{fig:realdata}
\end{figure}

From a policy or management perspective, we are interested in assessing the validity of the decision rule $T$ in this setting.
To this end, we apply the subset-identification framework from Section~\ref{sec:theory} to estimate effects at the cutoff.
By evaluating these effects, we draw conclusions for improving the treatment assignment $T$.

\subsection{Modeling assumptions and implications}
In this section, we derive modeling assumptions on the decision rule $D$ for the application setting,
combining insights from the observed data and the theory.
We assume that $D$ has an informational advantage over the initial decision rule
\(T = I_D \land I_Y\).
Specifically, we posit that this advantage arises from having more detailed information about the yield-improvement score.

The initial treatment assignment \(T\) is based only on the improvement estimates of every \(m\)-th 
item in the production lot, resulting in the score \(X_Y\).
In contrast, the final decision-maker has access to an overall yield-improvement estimate
\(X_E = X_Y + X_R\), where \(X_R\) reflects the contribution of items not included in \(X_Y\).
See Algorithm~\ref{algo:dgp} for details.

\begin{algorithm}[htb]
  \KwData{
    seed,
    lot-size \(n\),
    cutoff \(c\),
    measurement steps \(m\),
    yield criteria \(\mathcal{Y}\), 
    distance criteria \(\mathcal{D}\),
    operator-specific policy \(D_O\)
  }
  \KwResult{lot \(L\), scores \(X = (X_D, X_Y)\), assigned treatment \(T\), actual treatment \(D\), outcome \(Y\)}
  \(L \leftarrow (C_1, \ldots, C_n)\) generate a random production lot of \(n\) items\;
  \(X_D \leftarrow \mathcal{D}(L)\) calculate the distance to the target\;
  \(\hat{L}_A \leftarrow\) carry out an optimal rework step on \(L\)\;
  \(X_Y \leftarrow \mathcal{Y}(\hat{L}_A) - \mathcal{Y}(L)\) on every \(m\)-th item\;
  \(T \leftarrow \ind[X_D > c_D] \land \ind[X_Y > c_Y]\)\;
  \(X_E \leftarrow \mathcal{Y}(\hat{L}_A) - \mathcal{Y}(L)\) on every item\;
  \(D \leftarrow \ind[X_D > c_D] \land D_O(X_Y, \, X_E)\)\;
  \(L_A \leftarrow\) carry out a realistic (noisy) rework step on \(L\)\;
  \eIf{D}{
    \(Y \leftarrow \mathcal{Y}(L_A)\)
  }{
    \(Y \leftarrow \mathcal{Y}(L)\)
  }
  \caption{DGP for a knowledgeable operator}
  \label{algo:dgp}
\end{algorithm}

Implementing this algorithm allows us to benchmark different
\textit{operator} assumptions and compare them with the real-world case.
In semiconductor manufacturing, the term operator refers to the human decision-maker 
implementing \(D\).
Figure \ref{fig:ackn_caut} shows an example of data generated by this DGP.

Although the knowledgeable operator has, in principle, access to a better improvement estimate,
there are several ways in which this information may be applied.
In the following, we discuss special cases of the knowledgeable-operator assumption,
each of which (at least in theory) permits identification of the unit categories.
\begin{figure}[htb]
  \centering
  \includegraphics[width=\linewidth]{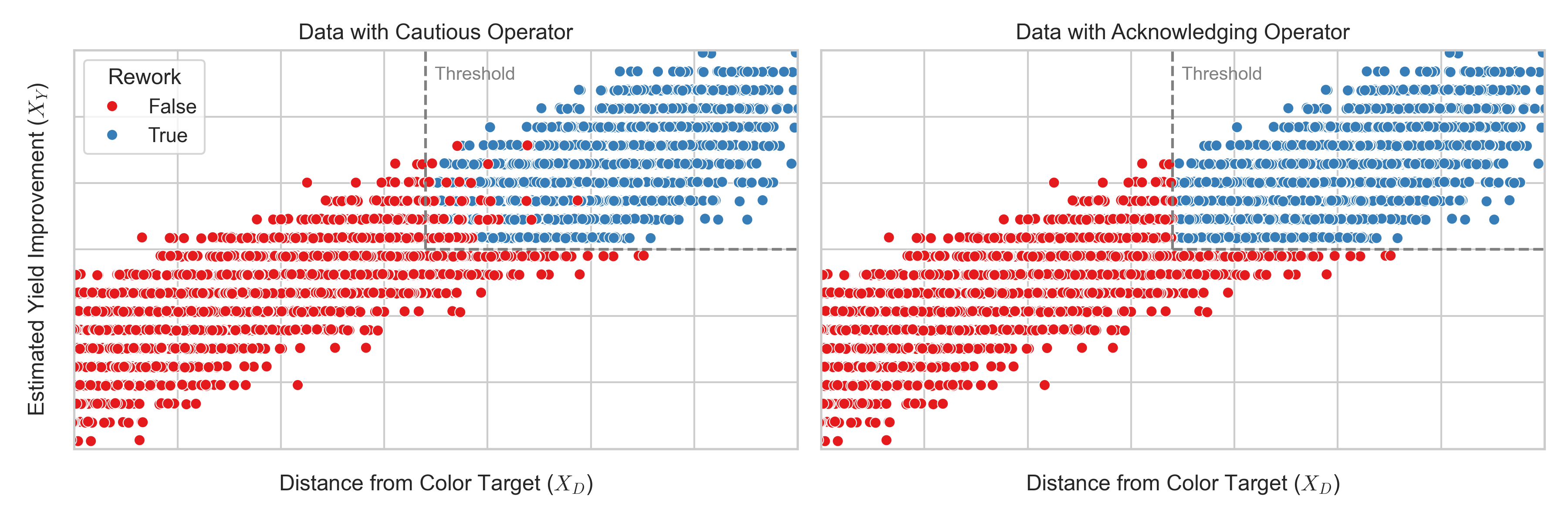}
  \caption{Comparison of data with cautious (left) and acknowledging (right) operator.}
  \label{fig:ackn_caut}
\end{figure}

\paragraph{The acknowledging operator}
In this scenario, the operator \(D\) accepts the treatment assignment prescribed by \(T\) despite having better information.
Thus, the final treatment assignment coincides with the initial decision rule:
\(D = T = I_D \land I_Y\).
This corresponds to the operator-specific policy in Algorithm~\ref{algo:dgp} with \(D_O = I_Y\).
Consequently, we obtain
\[
  \comp(T, D) = \mathcal{I}
\]
as well as
\[
  \comp(G, D) = \{i \,|\, X_{Y,i} > 0 \}
  \,\mbox{ and }\,
  \nevert(G, D) = \{i \,|\,  X_{Y,i} \leq 0 \}
\]
for \(G = I_D\), using Proposition~\ref{prop:operation_simple} or Example~\ref{ex:and_rule}.

\paragraph{The cautious operator}
Suppose the operator is particularly careful to avoid accidental degradation of the production lot caused by the optional rework step.
In cases where \(T\) recommends rework, the operator may override this recommendation.
Such overrides may be motivated either by better information or by additional constraints (e.g., time) that the intended treatment rule \(T\) does not capture.
In contrast, if \(T\) does not recommend rework, we assume the operator accepts this decision (e.g., out of concern about possible degradation).

We formalize this “negative overwrite” by introducing an additional score variable \(X_{op}\) and modeling \(D\) as:
\[
  D = (T \land I_{op}) \lor (\overline{T} \land 0) = T \land I_{op}.
\]
In particular, we set \(X_{op} = X_E\), which leads to
\[
  D_O(X \,|\, c) \defeq I_{X_Y} \land I_{X_E}.
\]
Thus, we have
\[
  \comp(T, D) = \{i \,|\, X_{E,i} > 0 \}
  \,\mbox{ and }\,
  \nevert(T, D) = \{i \,|\, X_{E,i} \leq 0 \}
\]
as well as
\[
  \comp(G, D) = \{i \,|\, X_{E,i} > 0, \, X_{Y,i} > 0 \}
  \,\mbox{ and }\,
  \nevert(G, D) = \{i \,|\, X_{E,i} \leq 0 \} \cup \{i \,|\, X_{Y,i} \leq 0 \}
\]
for the subrule \(G \defeq I_D\),
again using Proposition~\ref{prop:operation_simple} or Example~\ref{ex:and_rule}.

Employing Theorem~\ref{thm:identification:none_change_drop},
one can estimate the subset complier effect of \(G\), excluding the never-taker group
\(\Omega \defeq \{i \,|\, X_{Y,i} \leq 0\}\),
instead of the overall complier effect of \(G\).
Since we assume that \(X_E\) is known only to the operator,
nevertakers defined by \(X_{E,i} \leq 0\) cannot be excluded.
However, because the condition \(X_{Y,i} \leq 0\) applies globally,
the continuity and stability assumptions required for identification
are likely to hold in practice.

\paragraph{The reasonable operator}
Finally, consider an operator who always uses the additional information available.
In this case, the reasonable operator overwrites the intended assignment rule \(T\)
and bases the final decision solely on the full information \(X_E\) rather than on \(X_Y\):
\[
  D_O(X \,|\, c) = I_E(c).
\]
Accordingly, the final decision rule \(D\) can be expressed in relation to \(T\) as:
\[
  D(X \,|\, c) = I_D(c) \land \ind[X_Y + X_R > c_Y].
\]

A unit \(i\) belongs to \(\comp(T, D)\) if and only if \(X_{R,i} = 0\),
i.e., there is no improvement or degradation among the items excluded from \(X_Y\).
Now suppose \(X_{R,i} \neq 0\).
Then we can choose \(c_Y = 2 \max(X_{Y,i}, X_{R,i})\) and \(c_D < X_{D,i}\)
to obtain \(D(X_i \,|\, c) = T(X_i \,|\, c) = 0\).
However, there also exist cutoff values \(c \in \supp(T) = \mathbb{R}^2\)
for which \(D(X_i \,|\, c) \neq T(X_i \,|\, c)\).
For units with \(X_R \geq 0\), choose \(c_Y = X_Y\).
Then \(X_Y + X_R > c_Y\), so \(T(X_i \,|\, c) = 0\) and \(D(X_i \,|\, c) = 1\).
Otherwise, choose \(c = X_Y + X_R\).
Then \(c < X_Y\), so \(T(X_i \,|\, c) = 1\) and \(D(X_i \,|\, c) = 0\).

This shows that
\[
  \comp(T, D) = \{ i \,|\, X_{R,i} = 0 \},
\]
and that all other units fall into the indecisive category.

This model has several shortcomings.
First, it does not reflect the observed data: according to the model,
some treated items should not have received the initial assignment.
Second, the identification result depends on the absence of indecisive items.
In practice, small deviations in the improvement score (small \(X_R\))
would likely be considered compliant with the assignment \(T\),
which suggests the need for a more local definition of categories.

While this case is theoretically interesting (and useful as an example of the indecisive case),
our empirical investigation focuses on the two edge cases:
the acknowledging operator and the cautious operator,
as shown in Figure~\ref{fig:ackn_caut}.

\section{Simulation Study}
In this section, we present numerical results based on the semi-synthetic data-generating process (DGP) derived in the previous section.
The aim is to estimate the causal effect of rework decisions on the final yield in a neighborhood of the decision boundary.
To this end, we use different estimators at both decision thresholds separately and draw on the identification theorems from Section~\ref{sec:theory}.

\subsection{Set-Up}
We generate semi-synthetic data following Algorithm~\ref{algo:dgp}.
A Python implementation of the DGP and all estimators is publicly available.
The process is calibrated to match the characteristics of real production data.
We draw \(n = 10{,}000\) observations and repeat each experiment \(r = 250\) times.

We evaluate the cut-offs \(c_D\) and \(c_Y\) separately, estimating:
\begin{itemize}
  \item the complier effect of \(G \in \{I_Y, \, I_D\}\), as identified in Theorem~\ref{thm:identification};
  \item the subset complier effect of \(G\), conditional on its counterpart; and
  \item intent-to-treat (ITT) effects, with and without the subset conditioning.
\end{itemize}

Complier effects are estimated under a fuzzy design, whereas ITT effects are estimated under a sharp design.
Oracle values are obtained using local linear kernel regression on the differences in true potential outcomes.
Covariates consist of statistics describing the quality of individual items.

We consider a variety of estimators, including covariate-adjusted estimators, for complier and subset complier effects.
Table~\ref{tab:estimators} provides an overview of the estimators used in our analysis.

\begin{table}[!ht]
    \centering
    \scalebox{0.6}{
    \begin{tabular}{ll}
    \toprule
        \textbf{Method} & \textbf{Adjustment} \\ \midrule
        No Adjustment & - \\ 
        Conventional Adjustment & Linear \\ 
        RDFlex Lasso & Local Penalized Linear Estimators \\ 
        RDFlex Global Lasso & Global Penalized Linear Estimators \\ 
        RDFlex Boosting & Gradient Boosting Estimators \\ 
        RDFlex Stacking & Stacked Combination of Linear and Non-Linear Estimators \\ 
        \bottomrule
    \end{tabular}}
    \caption{Considered methods and estimators in this section.}
    \label{tab:estimators}
\end{table}

\subsection{Estimators}
The basic RDD estimator runs separate local linear regressions on each side of the cutoff:
\begin{equation}
	\hat{\tau}_{\text{base}}(h) = \sum_{i=1}^n w_i(h) Y_i,
\end{equation}
where the \(w_i(h)\) are local linear regression weights that depend on the data through the realizations of the running variable only, and \(h > 0\) is a bandwidth.

Under standard conditions (e.g., the running variable is continuously distributed and the bandwidth \(h\) tends to zero at an appropriate rate),
the estimator \(\hat{\tau}_{\text{base}}(h)\) is approximately normally distributed in large samples, with bias of order \(h^2\) and variance of order \((nh)^{-1}\).

A conventional extension is the covariate-adjusted estimator, which incorporates covariates linearly into the local regressions.
We also use modern RDD estimators with flexible covariate adjustment based on potentially nonlinear adjustment functions \(\eta\).
This estimator takes the form:
\begin{equation}
   \widehat{\tau}_{\text{RDFlex}}(h; \eta) = \sum_{i=1}^n w_i(h) M_i(\eta),
   \quad M_i(\eta) = Y_i - \eta(Z_i),
\end{equation}
where \(\eta\) denotes the influence of \(Z\) on the outcome \(Y\), estimated using machine learning methods.

\subsection{Cautious Operator}\label{sec:sim}
Figure~\ref{fig:res_cautious_comparison} presents side-by-side results for the estimators at \(c_D\) and \(c_Y\) in the case of the cautious operator.
The intent-to-treat (ITT) oracles are closer to zero than the complier effects because they include individuals who are nevertakers with respect to each cutoff rule.
As shown in Figure~\ref{fig:res_cautious_dis}, for \(I_D\) we estimate an overall negative effect, although it is not statistically significant at the 95\% level.

The subset effect for the fuzzy case exhibits a smaller bias, since the proportion of nevertakers in the estimation sample is lower.
The estimated effect remains unchanged because only nevertakers, but no compliers, were removed.
For the ITT estimator, a higher proportion of compliers in the subsample increases the estimated effect of treatment rule \(G\).
The subset estimators have a comparable variance (see Table~\ref{tab:cautious_dis}), with coverage appearing slightly more credible overall.
In general, covariate adjustment reduces standard errors, especially for the sharp estimators.

The estimates for \(I_Y\) are small and positive.
The fuzzy estimator on the full data has a high standard error, which increases further with ML adjustment.
This may be due to a small jump in treatment probability in the full data, destabilizing the ML estimate.
In contrast, the subset estimator along this axis removes more observations within the estimation bandwidth, thereby reducing variance.
Additional results can be found in Appendix~\ref{app:addres}.

\begin{figure}[htb]
    \centering
    \begin{subfigure}[b]{0.48\linewidth}
        \includegraphics[width=\linewidth]{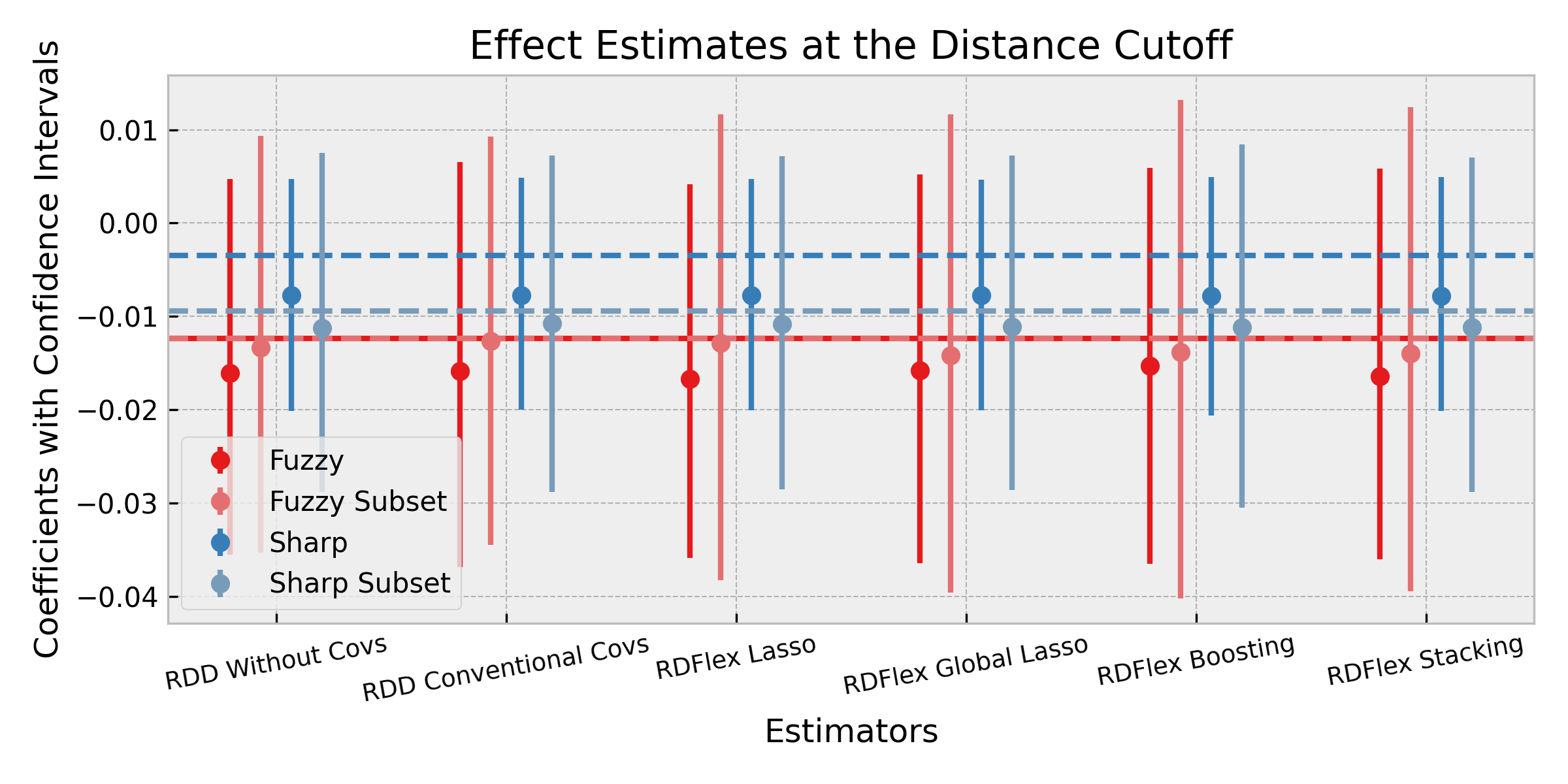}
        \caption{Results for fuzzy and sharp estimators at the cutoff \(c_D\) both on the full sample and on the subset.}
        \label{fig:res_cautious_dis}
    \end{subfigure}
    \hspace{0.02\linewidth}
    \begin{subfigure}[b]{0.48\linewidth}
        \includegraphics[width=\linewidth]{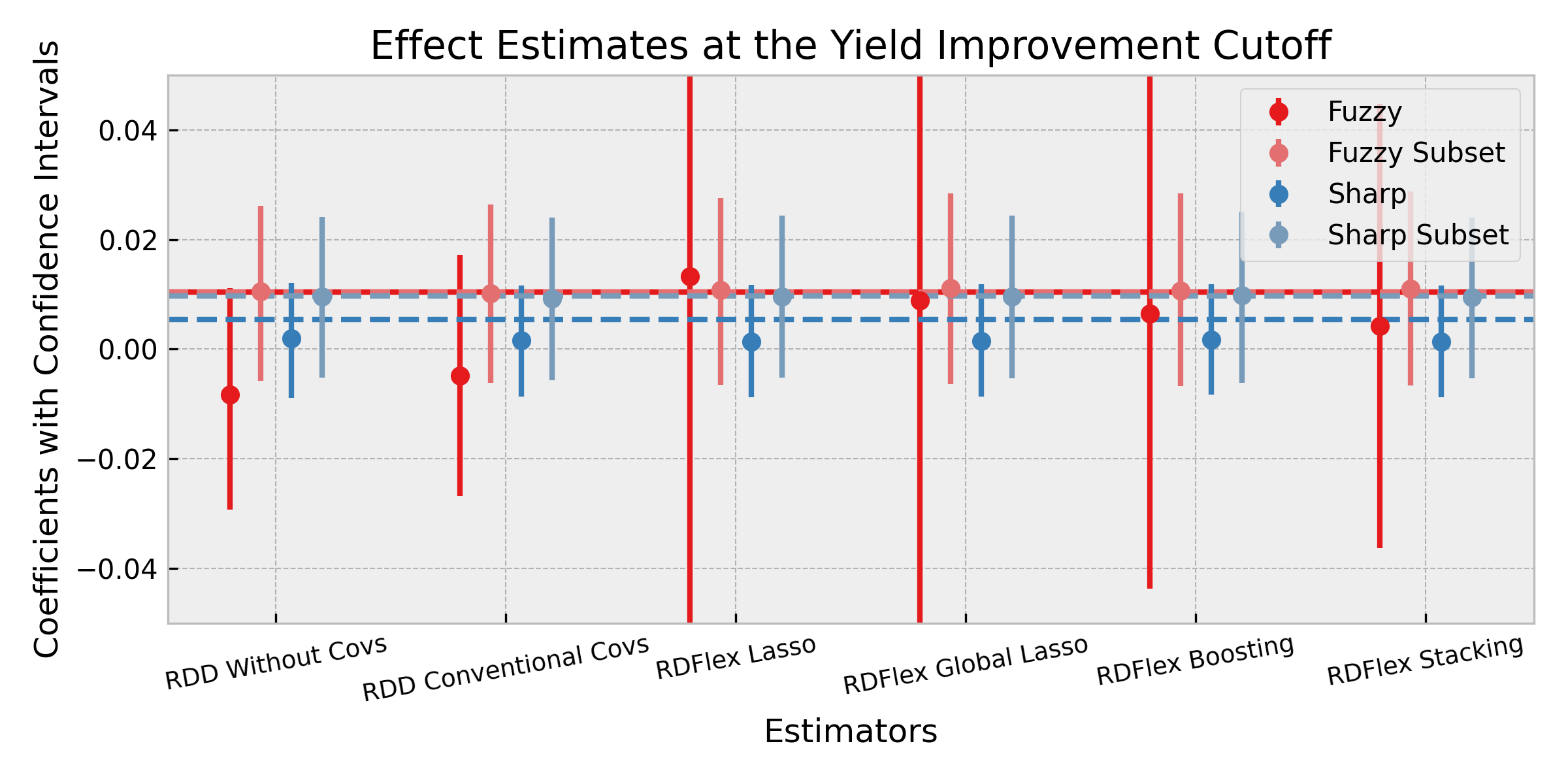}
        \caption{Results for fuzzy and sharp estimators at the cutoff \(c_Y\) both on the full sample and on the subset.}
        \label{fig:res_cautious_yie}
    \end{subfigure}
    \caption{
      Comparison of median estimates with median 95\% confidence intervals at \(c_D\) and \(c_Y\) on the simulated data with the ``cautious'' operator assumption.
      Different colors depict complier effect (fuzzy) as well as intent-to-treat effect (sharp).
      The subset estimator excludes nevertakers according to Theorem~\ref{thm:identification:none_change_drop}.
    }
    \label{fig:res_cautious_comparison}
\end{figure}

\begin{table}[htb]
\centering
\scalebox{0.5}{%
\begin{tabular}{llrrrrrrr}
\toprule
 &  & Mean Bias & s.e. & Coverage & RMSE left & Log loss left & RMSE right & Log loss right \\
setting & method &  &  &  &  &  &  &  \\
\midrule
\multirow[t]{6}{*}{Fuzzy} & RDD Conventional Covs & -0.0025 & 0.0111 & 0.9912 &  &  &  &  \\
 & RDD Without Covs & -0.0027 & 0.0109 & 0.9735 &  &  &  &  \\
 & RDFlex Boosting & -0.0025 & 0.0095 & 0.9912 & 0.0245 & 0.0532 & 0.0923 & 0.4996 \\
 & RDFlex Global Lasso & -0.0028 & 0.0097 & 0.9823 & 0.0399 & 0.0987 & 0.0897 & 0.6070 \\
 & RDFlex Lasso & -0.0030 & 0.0092 & 0.9823 & 0.0242 & 0.0787 & 0.0897 & 0.5645 \\
 & RDFlex Stacking & -0.0029 & 0.0092 & 0.9823 & 0.0243 & 0.0169 & 0.0898 & 0.4938 \\
\cline{1-9}
\multirow[t]{6}{*}{Fuzzy Subset} & RDD Conventional Covs & -0.0001 & 0.0117 & 0.9609 &  &  &  &  \\
 & RDD Without Covs & -0.0003 & 0.0118 & 0.9478 &  &  &  &  \\
 & RDFlex Boosting & -0.0009 & 0.0116 & 0.9652 & 0.0305 & 0.0741 & 0.1050 & 0.3786 \\
 & RDFlex Global Lasso & -0.0004 & 0.0111 & 0.9522 & 0.0259 & 0.1961 & 0.1012 & 0.3665 \\
 & RDFlex Lasso & -0.0002 & 0.0110 & 0.9435 & 0.0249 & 0.4230 & 0.1010 & 0.5140 \\
 & RDFlex Stacking & -0.0004 & 0.0113 & 0.9522 & 0.0258 & 0.0329 & 0.1011 & 0.3580 \\
\cline{1-9}
\multirow[t]{6}{*}{Sharp} & RDD Conventional Covs & -0.0040 & 0.0063 & 0.9120 &  &  &  &  \\
 & RDD Without Covs & -0.0040 & 0.0063 & 0.9240 &  &  &  &  \\
 & RDFlex Boosting & -0.0040 & 0.0056 & 0.9080 & 0.0247 &  & 0.0918 &  \\
 & RDFlex Global Lasso & -0.0040 & 0.0054 & 0.9240 & 0.0401 &  & 0.0888 &  \\
 & RDFlex Lasso & -0.0040 & 0.0055 & 0.9200 & 0.0242 &  & 0.0886 &  \\
 & RDFlex Stacking & -0.0040 & 0.0055 & 0.9240 & 0.0243 &  & 0.0888 &  \\
\cline{1-9}
\multirow[t]{6}{*}{Sharp Subset} & RDD Conventional Covs & -0.0014 & 0.0094 & 0.9680 &  &  &  &  \\
 & RDD Without Covs & -0.0014 & 0.0095 & 0.9600 &  &  &  &  \\
 & RDFlex Boosting & -0.0016 & 0.0085 & 0.9760 & 0.0280 &  & 0.1045 &  \\
 & RDFlex Global Lasso & -0.0014 & 0.0079 & 0.9640 & 0.0253 &  & 0.1006 &  \\
 & RDFlex Lasso & -0.0014 & 0.0079 & 0.9560 & 0.0246 &  & 0.1004 &  \\
 & RDFlex Stacking & -0.0015 & 0.0080 & 0.9600 & 0.0256 &  & 0.1006 &  \\
\cline{1-9}
\bottomrule
\end{tabular}

}
\caption{Mean bias, standard error, coverage, and first-stage prediction quality of different estimators in estimation at \(c_D\).}
\label{tab:cautious_dis}
\end{table}

\paragraph{Two-dimensional Estimators}
We also consider two-dimensional estimators proposed in previous literature.
We compare the ``binding-the-score'' approach, which normalizes and combines \(X_D\) and \(X_Y\) into a single score,
and the ``Euclidean distance'' approach, which computes the shortest distance of each observation to a predetermined point (here, a point on \(c_D\)).

\begin{figure}[htb]
\centering
\includegraphics[width=0.7\linewidth]{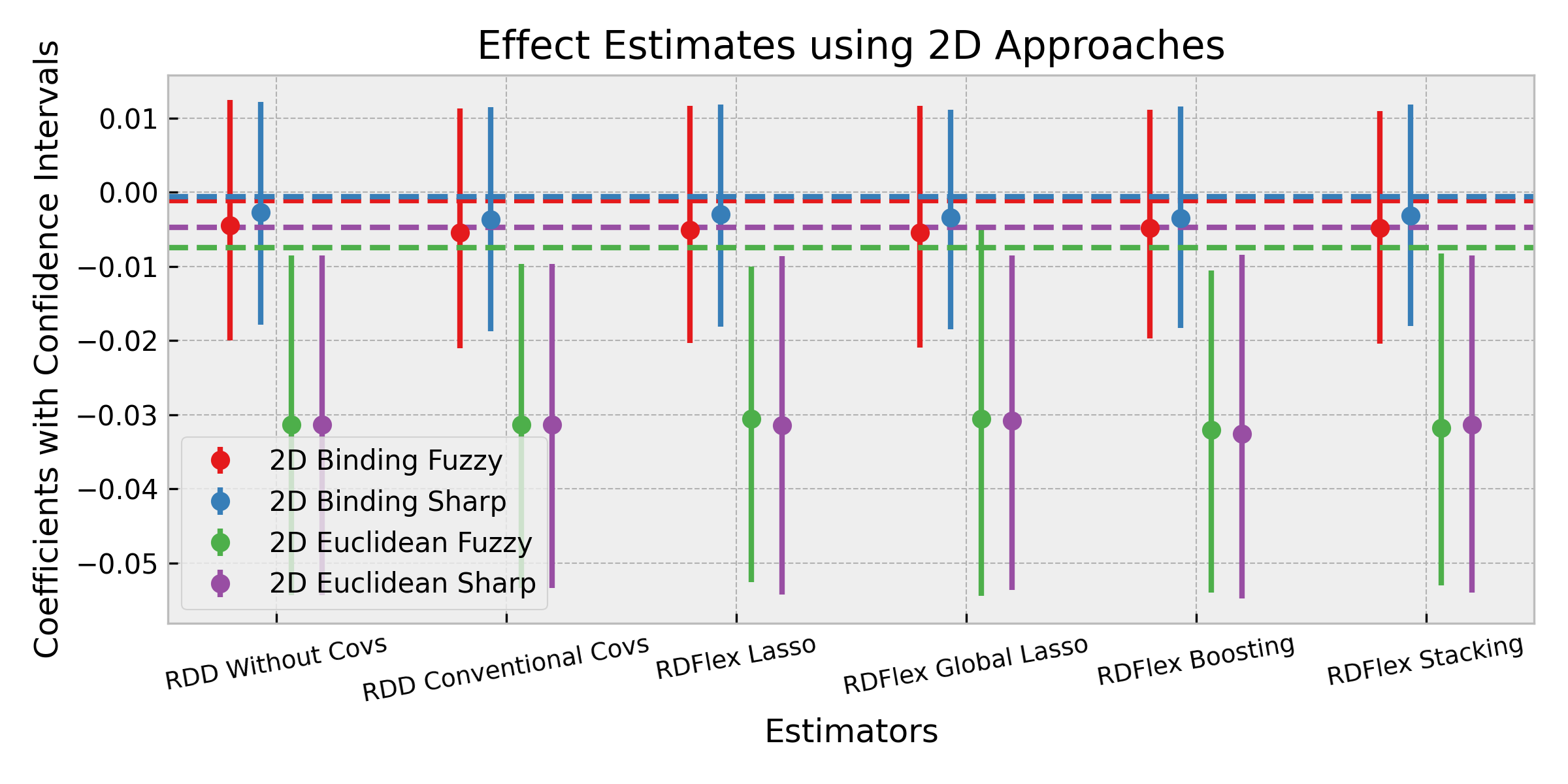}
\caption{Median estimates with median 95\% confidence intervals for estimators that use two dimensions at once.}
\label{fig:2dres}
\end{figure}

The results in Figure~\ref{fig:2dres} show oracles very close to zero, indicating that when collapsing both dimensions of the decision rule into one,
it is no longer possible to obtain insights into the individual rules.
Interestingly, the estimator based on the Euclidean distance exhibits a significant bias across all methods.

In our specific setting, there is no straightforward interpretation of the RDD effect estimate at a two-dimensional threshold,
which contrasts with previous applications where score components are more directly comparable, such as test scores or geographical distances.

\subsection{Acknowledging Operator}
We now compare the results from the previous section to those for the ``acknowledging'' operator,
who accepts the decision of the system without using additional domain knowledge.
In this setting, the data follow a sharp two-dimensional MRD.
The comparison of the effects at \(c_Y\) and \(c_D\) is shown in Figure~\ref{fig:res_ack_comparison}.

\begin{figure}[htb]
    \centering
    \begin{subfigure}[b]{0.48\linewidth}
        \includegraphics[width=\linewidth]{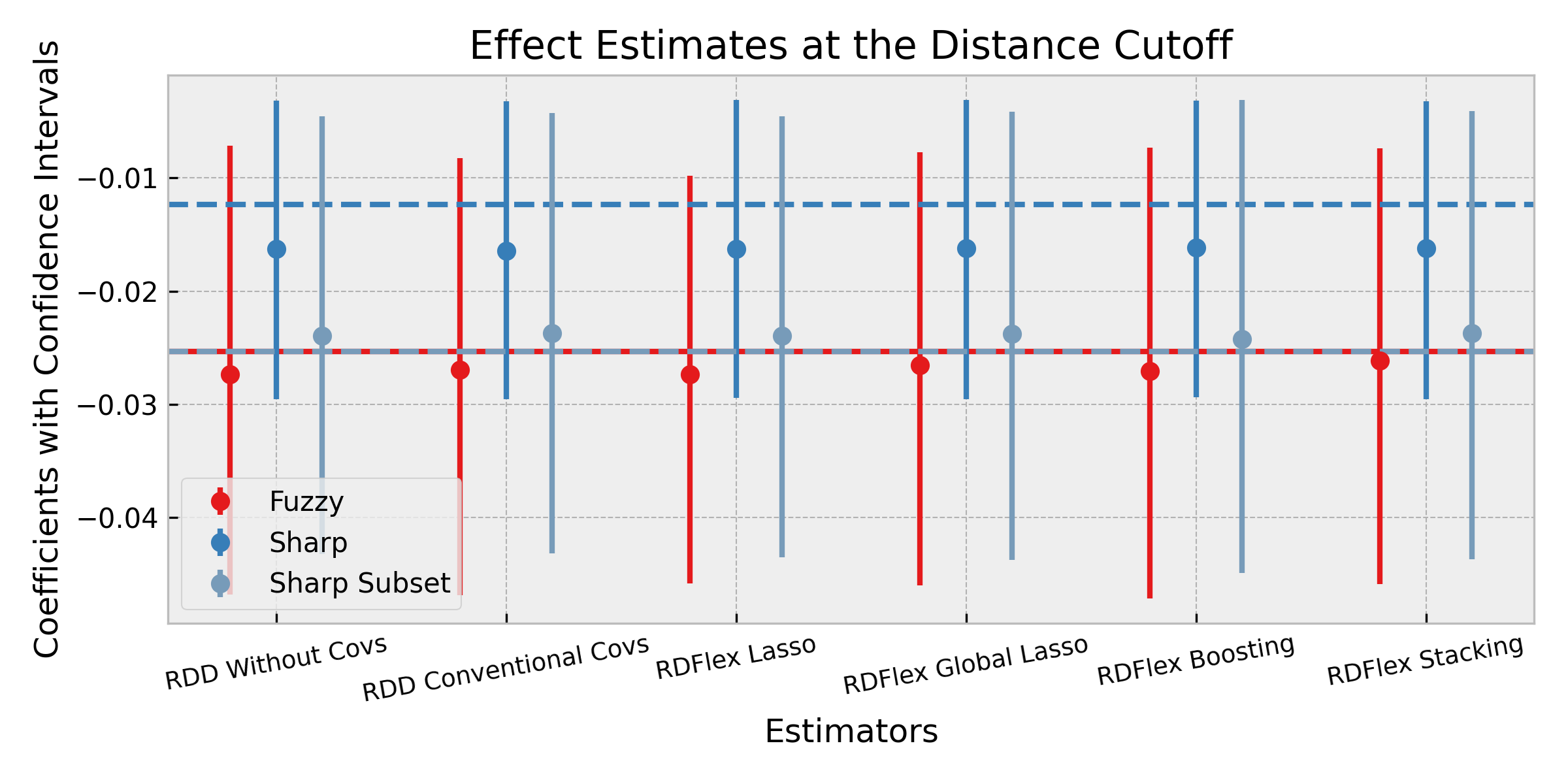}
        \caption{Results for fuzzy and sharp estimators at the cutoff \(c_D\) both on the full sample and on the subset.}
        \label{fig:res_ack_dis}
    \end{subfigure}
    \hspace{0.02\linewidth}
    \begin{subfigure}[b]{0.48\linewidth}
        \includegraphics[width=\linewidth]{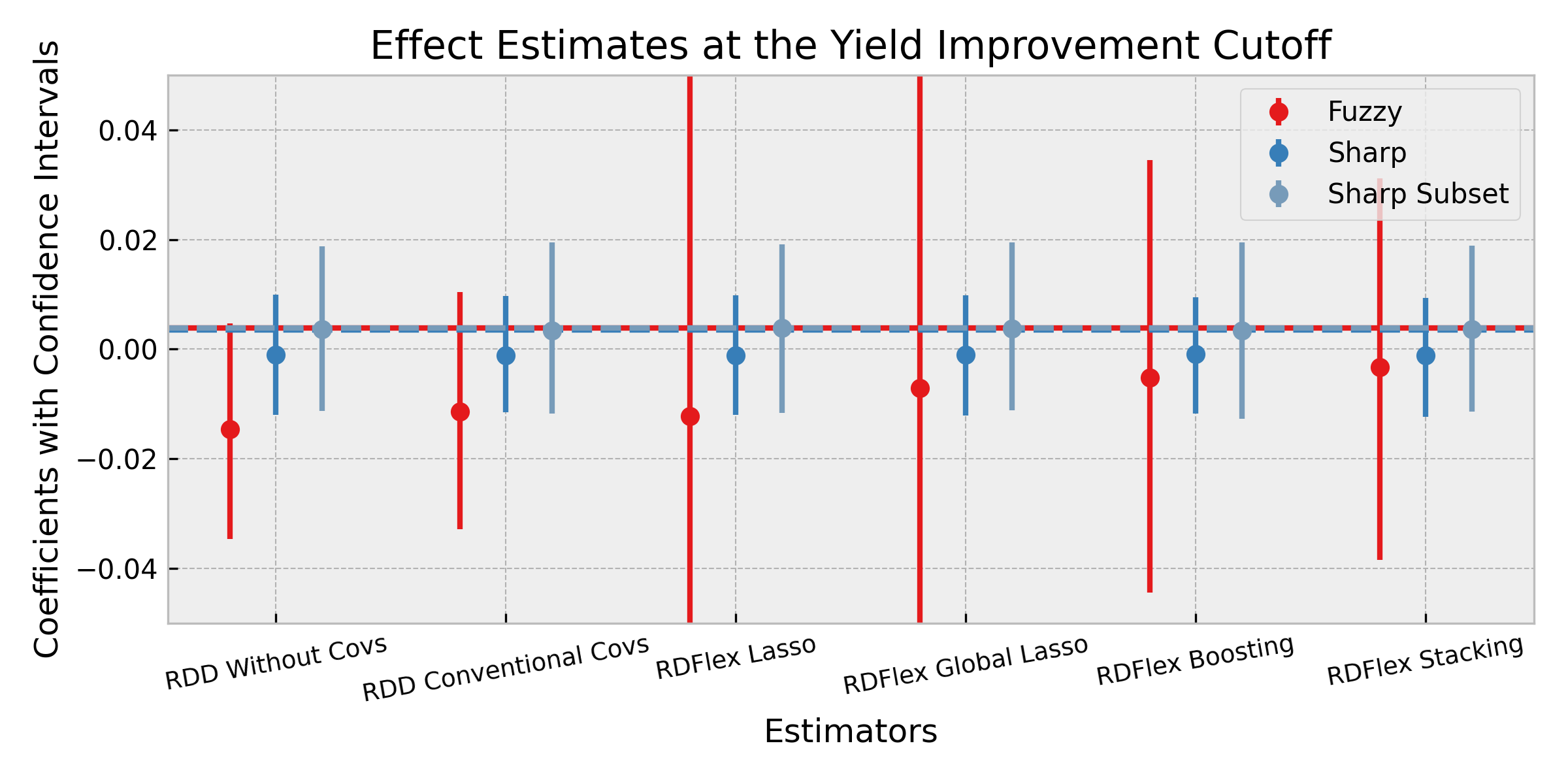}
        \caption{Results for fuzzy and sharp estimators at the cutoff \(c_Y\) both on the full sample and on the subset.}
        \label{fig:res_ack_yie}
    \end{subfigure}
    \caption{
      Comparison of median estimates with median 95\% confidence intervals at \(c_D\) and \(c_Y\) on the simulated data with the ``acknowledging'' operator assumption.
      Different colors depict complier effect (fuzzy) as well as intent-to-treat effect (sharp).
      The subset estimator excludes nevertakers according to Theorem~\ref{thm:identification:none_change_drop}.
    }
    \label{fig:res_ack_comparison}
\end{figure}

For \(c_D\), the estimated effects are smaller than under the cautious operator.
Lots that the cautious operator withheld from rework due to concerns about degradation are reworked in this case.
This leads to a larger overall negative effect and suggests that the threshold should be moved even further than in the cautious operator case.
The sharp subset estimator yields the same effect as the fuzzy estimator on the full data,
because subsetting according to \(I_Y\) produces a sharp design in which all individuals comply with threshold movements.
Across estimators, the subset versions display slightly smaller bias and standard errors.
In contrast, the sharp estimator on the full sample, which corresponds to an ITT estimator,
produces an effect closer to zero because it includes nevertakers.

At the cutoff \(c_Y\), the operator’s behavior has little impact on effect estimates.
This can be explained by the fact that relatively few individuals in the neighborhood of the cutoff are affected.

\section{Empirical Results}
This section presents results from estimation on the real data, which consist of \(n=9{,}103\) observations from the production system.
In addition to the covariates used above, we include shop-floor workload.
Our primary interest in the real system is evaluating the rule \(I_D\).

As shown in Table~\ref{tab:real_data}, we estimate a positive effect.
The subset estimates display lower variance because uncertainty from nevertakers (according to \(I_Y\)) is removed.
For the ITT estimators, ML adjustment improves variance, whereas for the fuzzy estimators no clear improvement is observed.

\begin{table}[htb]
    \centering
    \scalebox{0.5}{%
    \begin{tabular}{llrrrrrrrrr}
\toprule
             &                 &   Coef &   s.e. &  CI 2.5\% &  CI 97.5\% &  RMSE left &  Log loss left &  RMSE right &  Log loss right &  \% s.e. change \\
setting & method &        &        &          &           &            &                &             &                 &                 \\
\midrule
Fuzzy & RDD Conventional Covs & 0.1031 & 0.0508 &   0.0035 &    0.2028 &            &                &             &                 &         22.2532 \\
             & RDD Without Covs & 0.0387 & 0.0416 &  -0.0428 &    0.1202 &            &                &             &                 &          0.0000 \\
             & RDFlex Boosting & 0.0741 & 0.0524 &  -0.0400 &    0.1882 &     0.1432 &         0.0222 &      0.1594 &          0.6264 &         25.9991 \\
             & RDFlex Global Lasso & 0.0519 & 0.0461 &  -0.0484 &    0.1522 &     0.1597 &         0.0108 &      0.1510 &          0.6653 &         10.8590 \\
             & RDFlex Local Lasso & 0.0999 & 0.0470 &  -0.0004 &    0.2003 &     0.1452 &         0.0141 &      0.1523 &          0.6142 &         12.9249 \\
             & RDFlex Stacking & 0.1097 & 0.0500 &   0.0023 &    0.2171 &     0.1429 &         0.0188 &      0.1570 &          0.6322 &         20.1333 \\
\cline{1-9}
Fuzzy Subset & RDD Conventional Covs & 0.0630 & 0.0426 &  -0.0205 &    0.1466 &            &                &             &                 &         36.5324 \\
             & RDD Without Covs & 0.0369 & 0.0312 &  -0.0243 &    0.0981 &            &                &             &                 &          0.0000 \\
             & RDFlex Boosting & 0.0581 & 0.0332 &  -0.0163 &    0.1325 &     0.1423 &         0.0200 &      0.1562 &          0.5158 &          6.4133 \\
             & RDFlex Global Lasso & 0.0270 & 0.0375 &  -0.0537 &    0.1076 &     0.1548 &         0.0113 &      0.1520 &          0.5731 &         20.2628 \\
             & RDFlex Local Lasso & 0.0713 & 0.0304 &   0.0031 &    0.1395 &     0.1410 &         0.0209 &      0.1507 &          0.5424 &         -2.5667 \\
             & RDFlex Stacking & 0.0768 & 0.0319 &   0.0048 &    0.1487 &     0.1394 &         0.0121 &      0.1530 &          0.5018 &          2.1538 \\
\cline{1-9}
Sharp & RDD Conventional Covs & 0.0393 & 0.0208 &  -0.0015 &    0.0802 &            &                &             &                 &          2.2467 \\
             & RDD Without Covs & 0.0165 & 0.0204 &  -0.0235 &    0.0564 &            &                &             &                 &          0.0000 \\
             & RDFlex Boosting & 0.0311 & 0.0176 &  -0.0093 &    0.0715 &     0.1417 &                &      0.1575 &                 &        -13.5386 \\
             & RDFlex Global Lasso & 0.0200 & 0.0171 &  -0.0193 &    0.0593 &     0.1545 &                &      0.1512 &                 &        -16.3636 \\
             & RDFlex Local Lasso & 0.0357 & 0.0167 &  -0.0026 &    0.0740 &     0.1425 &                &      0.1513 &                 &        -17.9717 \\
             & RDFlex Stacking & 0.0367 & 0.0168 &  -0.0018 &    0.0753 &     0.1408 &                &      0.1530 &                 &        -17.5203 \\
\cline{1-9}
Sharp Subset & RDD Conventional Covs & 0.0362 & 0.0205 &  -0.0040 &    0.0764 &            &                &             &                 &         -2.3944 \\
             & RDD Without Covs & 0.0098 & 0.0210 &  -0.0314 &    0.0509 &            &                &             &                 &          0.0000 \\
             & RDFlex Boosting & 0.0222 & 0.0196 &  -0.0225 &    0.0669 &     0.1424 &                &      0.1590 &                 &         -6.6630 \\
             & RDFlex Global Lasso & 0.0102 & 0.0179 &  -0.0306 &    0.0509 &     0.1533 &                &      0.1518 &                 &        -14.7987 \\
             & RDFlex Local Lasso & 0.0364 & 0.0184 &  -0.0031 &    0.0758 &     0.1408 &                &      0.1525 &                 &        -12.2509 \\
             & RDFlex Stacking & 0.0307 & 0.0181 &  -0.0093 &    0.0708 &     0.1395 &                &      0.1540 &                 &        -14.0691 \\
\cline{1-9}
\bottomrule
\end{tabular}

    }
    \caption{Coefficient, standard error, and quality of fit for the different estimators for \(c_D\) in the real data.}
    \label{tab:real_data}
\end{table}

The comparison in Figure~\ref{fig:rdplots} highlights differences between the semi-synthetic and real data.
The semi-synthetic results (panel a) suggest that the rework cutoff is miscalibrated, leading to negative estimated effects,
whereas the real data (panel b) show a positive jump at the current cutoff, consistent with a beneficial rework threshold.

\begin{figure}[htb]
    \centering
    \begin{subfigure}[b]{0.4\linewidth}
        \includegraphics[width=\linewidth]{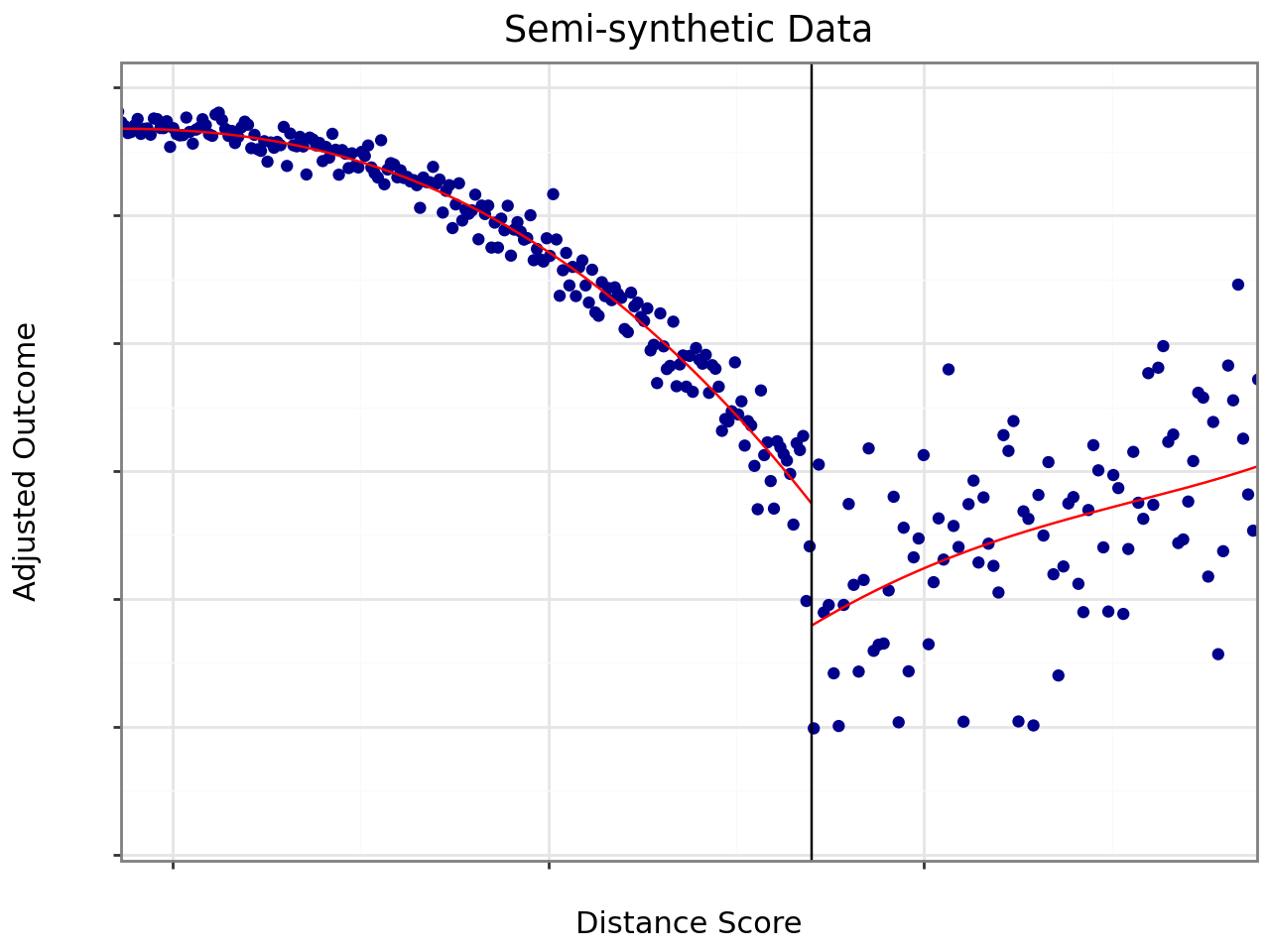}
        \caption{RDD plot for the semi-synthetic process.}
    \end{subfigure}
    \hspace{0.16\linewidth}
    \begin{subfigure}[b]{0.4\linewidth}
        \includegraphics[width=\linewidth]{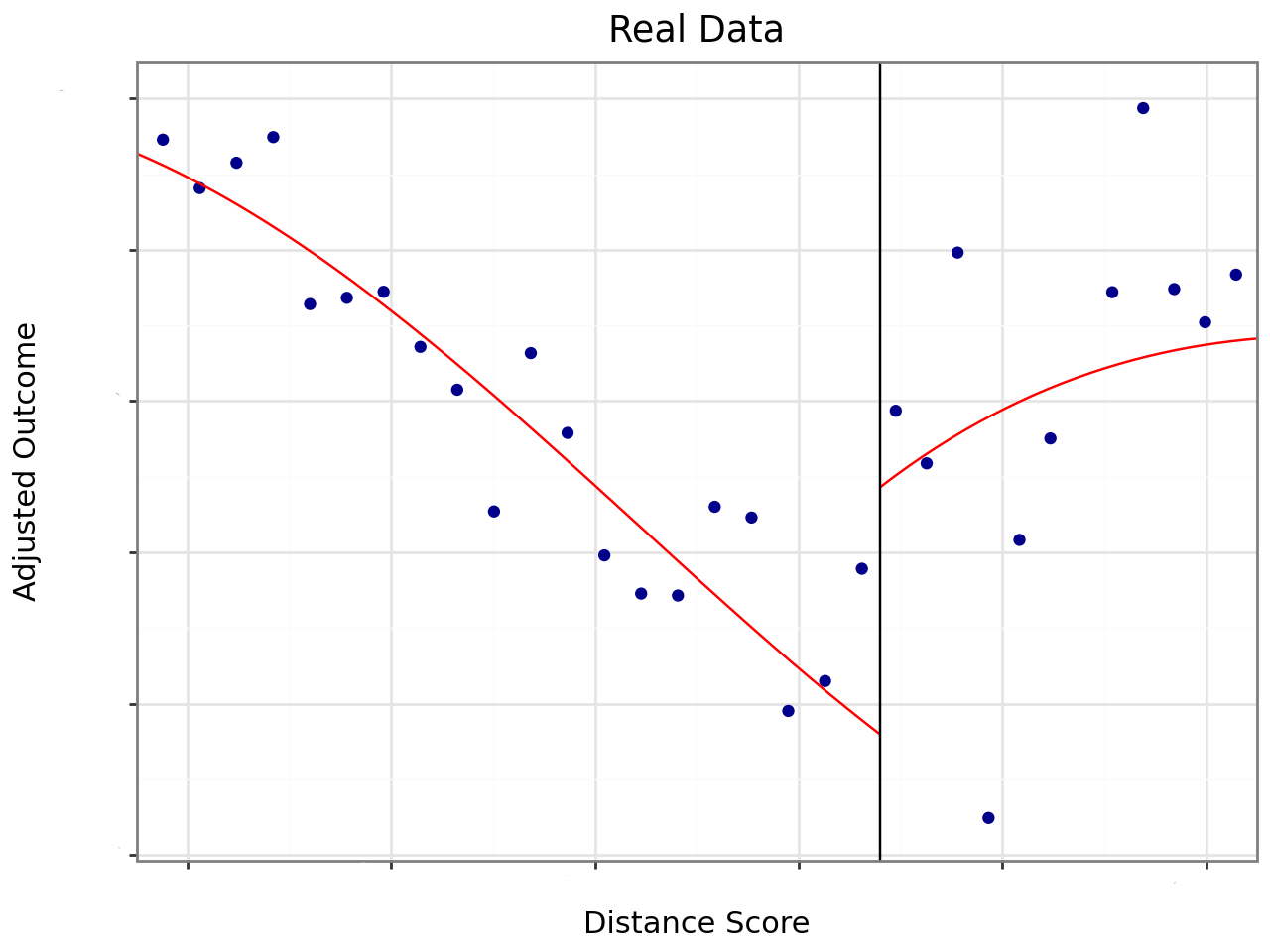}
        \caption{RDD plot for the real production data.}
    \end{subfigure}
    \caption{Comparison of local linear regressions around the cutoff for the semi-synthetic (a) and real (b) data. The opposite sign of the estimated effects can be explained by calibration of the semi-synthetic process.}
    \label{fig:rdplots}
\end{figure}

\paragraph{Validation Check.}
We validate the application results using a pseudo-cutoff test (see \cite{Cattaneo2019}).
We compare the estimated ITT effect at the true cutoff \(c_D\) with estimates at a hypothetical lower and higher cutoff.

As shown in Table~\ref{tab:validation}, both pseudo cutoffs produce effects close to zero with wide confidence intervals,
whereas the effect at the true cutoff is nearly significant at the 95\% level.

\begin{table}[H]
\centering
    \scalebox{0.6}{
\begin{tabular}{lrrr}
\toprule
       Cutoff &     Coef &  2.5 \% CI &  97.5 \% CI \\
\midrule
 lower pseudo & 0.002919 & -0.014796 &   0.020634 \\
         real & 0.039332 & -0.001524 &   0.080188 \\
higher pseudo & 0.041137 & -0.015333 &   0.097606 \\
\bottomrule
\end{tabular}}
\caption{Coefficient and confidence interval of the sharp real-data estimation at pseudo cutoffs for \(c_D\).}
\label{tab:validation}
\end{table}

\section{Conclusion}
This paper improves the applicability of regression discontinuity design (RDD) in operations and management.
The multi-score RDD framework, combined with modern machine-learning-based adjustment, provides new opportunities to evaluate complex decision settings and policies.
In settings with strict decision boundaries, policy-makers face little to no overlap in treatments near the cutoff, limiting the scope for ``what if'' reasoning.

Our theoretical insights rest on three observations.
First, analyzing general assignment rules requires extending the notion of ``binary compliance'' to the broader concept of ``compliance with an assignment mechanism.'' 
In principle, all possible inputs to the mechanism must be considered when assessing compliance.
Second, compliance is always defined relative to the rules under consideration.
Changing the assignment mechanism can change the observed behavior of units.
This relativity should be reflected in theory, as knowledge about specific parts of the mechanism or the final decision rule makes it possible to identify noncompliance.
Third, in multiple dimensions, the complier effect at the cutoff should not depend on the direction from which the cutoff is approached, a property that future estimation methods can exploit. 

Building on these observations, we generalized existing unit behavior types (compliers, nevertakers, alwaystakers) to multi-dimensional cutoff rules and introduced a new category of indecisive units.
We derived rules describing how behavior types evolve when cutoff rules are decomposed into subrules.
We also proved an identification result for the cutoff complier effect in MRD settings and established conditions under which this result remains valid when subsets of alwaystakers and nevertakers are removed, leading to the subset complier effect at the cutoff.

We validated our theoretical results using simulation studies with a semi-synthetic DGP and a real-world application in manufacturing.
These analyses show that subrule-specific estimation can improve decision-making.
In particular, the results confirm that removing identifiable nevertakers and/or alwaystakers makes effect estimation more precise, under conditions specified by the theory.
When combined with flexible ML-based adjustments, variance can be further reduced without introducing bias.
We recommend that practitioners estimate both the complier and the subset complier effect at the cutoff whenever nevertakers or alwaystakers can be identified in practice, and retrospectively validate the required assumptions.
Importantly, the empirical findings underline asymmetries in treatment effects across different cutoff dimensions: the effect of the overall assignment rule generally differs from the effects of its subrules.
Nevertheless, careful evaluation of each score component can provide valuable insights.

We expanded MRD theory by introducing tools to analyze complex cutoff rules and their corresponding subrules, expressed in terms of well-defined unit behavior types.
Although our empirical study specifically bridges econometric RDD methods and industrial process control, the framework is broadly applicable to other complex decision settings, such as multi-criteria loan approval in financial institutions or public health resource allocation where eligibility depends on several risk factors.
Future work could explore the relationship between main rules and subrules in greater depth, study the stability of unit types, or extend the multi-dimensional RDD to hierarchical decision structures across multiple stages.
Developing estimators tailored to the identification theorems also represents a promising direction.

\newpage

\appendix
\section{Proofs}\label{appendix:proofs}
\subsection{Unit categorization}
\lemconstsupp*
Before proving the above statement let us make some simple observations. Let
    \[S(T\,|\,X) \defeq \{ k \in \Nn_{\leq K} \,|\, \lambda \in \Rn \,:\, T_i(X\,|\,0) \neq T(X\,|\,0+\lambda e_k)\}\]
denote the set of local change directions of \(X\). Then one has the following properties:
\begin{enumerate}
  \item \(S(T\,|\,X)\) depends only upon the quadrant of \(X\) since within a quadrant the \(I_k\)`s and \(\overline{I}_k\)`s are constant.

  \item  \(S(T) = \bigcup_{X \in \Rn^K} S(T\,|\, X)\) 

    Since \(T(0 \,|\, c + \lambda e_k) \neq T(0 \,|\, c)\) is equivalent to \(T(-c \,|\, 0 + \lambda e_k) \neq T(-c \,|\, 0)\).

  \item \(S(T\,|\,X)\) can be empty.

    For example let \(T = (I_1 \land I_2) \lor I_3\) and \(X \in \Rn^3_{>0}\). Then \(S(T\,|\, X) = \emptyset\),
    since \(T\) is true as long as there are at least two indicators that are true.
    Changing the cutoff in only one direction \(k\) does not affect \(T(0 + \lambda e_k)\).
    This example can be extended such that changes in multiple directions do not affect \(T\).
\end{enumerate}
\begin{proof}
One direction is clearly trivial. For the other suppose \(T\) is not constant.
Then one has \(0 = T(X\,|\, 0) \neq T(\hat{X}\,|\,0) = 1\) for appropriate \(X,\, \hat{X} \in \Rn^K\).
Define the sequence
\[X^k \defeq X + \sum_{j=1}^k \sprodk{\hat{X}-X}{e_j}e_j \quad 0 \leq k \leq K\]
for \(0 \leq k \leq K\).
Then \(X^K = \hat{X}, \, X^0 = X\) and \(X^k - X^{k-1} = \lambda_k e_k\) with \(\lambda_k \defeq \sprodk{\hat{X}-X}{e_k}\).
We assume that \(S(T\,|\, X^k) = \emptyset\) for all \(k\).
Then
\[ T(X^{k-1}\,|\, 0) = T(X^{k-1}\,|\, 0 - \lambda_k e_k) = T(X^k \,|\, 0) \]
and thus per induction \(0 = T(X\,|\, 0) = T(\hat{X}\,|\, 0) = 1\). Which contradicts the assumption.
Thus, there exists some \(k\) with \(S(T\,|\, X^k) \neq \emptyset\).
Which shows that \(S(T) \neq \emptyset\).
\Halmos
\end{proof}
As a general assumption for the next statements we require that \(T\) is not trivial, that is \(\supp(T) \neq \{0\}\).
\propdecomposition*
\begin{proof}
Using equations \ref{eq:transf_cutoff} one can see that \(X + \lambda Y \in N^T\) for \(X, Y \in N^T\) and \(\lambda \in \Rn_{>0}\).
Further \(T(-X \,|\, c) = T(0 \,|\, c + X) = T(X \,|\, c + X) = T(X-X \,|\, c) = T(0 \,|\, c)\) for \(X \in N^T\) and \(c \in \Rn^K\).
Thus, \(N^T\) is a linear subspace of \(\Rn^K\).
Let \(k \in S(T)\) then there exists \(\lambda \in \Rn,\, c \in \Rn^K\) such that
\(T(0 \,|\, c + \lambda e_k) \neq T(0 \,|\, c)\).
This means that \(-\lambda e_k \notin N\). 
Since \(N^T\) is a linear space this means that \(e_k\) can also not be in \(N\).
Now let \(k \notin S(T)\). 
Then \(T(0 \,|\, c) = T(0 \,|\, c - e_k) = T(e_k \,|\, c)\) for all \(c \in \Rn^K\) and thus \(e_k \in N^T\).
This shows that \(\supp(T) \cap N^T = \{0\}\).
Let \(X^T \defeq \sum_{k \in S(T)} \sprodk{X}{e_k} e_k\) be the projection on to the subspace \(\supp(T)\)
and \(X^{\perp T} \defeq X - X^T\).
Suppose there exists some \(c\) such that \(T(X^{\perp T} \,|\, c) \neq T(0 \,|\, c)\).
Let 
\[ c^k \defeq c - \hat{X} + \sum_{j=1}^k \lambda_j e_j \quad 0 \leq k \leq K\]
with \(\lambda_j \defeq \sprodk{X^{\perp T}}{e_j}\).
Then \(c^0 = c - X^{\perp T}, \, c^K = c\) and \(c^k = c^{k-1} + \lambda_k e_k\) holds.
Further, one has
\[T(0\,|\, c^k) = T(0 \,|\, c^{k-1} + \lambda_k e_k) = T(0 \,|\, c^{k-1})\]
for all \(1 \leq k \leq K\).
Otherwise, \(k \subset S(T)\) and thus \(\lambda_k = 0\) by definition of \(X^{\perp T}\).
Using induction we derive \(T(X^{\perp T} \,|\, c) = T(0\,|\, c^0) = T(0 \,|\, c^K)\) which contradicts the assumption.
For uniqueness suppose there exists another decomposition \(X = Z^T + Z^{\perp T}\) with the above properties.
Then \(Z^T - X^T = Z^{\perp T} - X^{\perp T} \in \supp(T) \cap N^T = \{0 \}\).
\Halmos
\end{proof}
\proequivcutoffrules*
\begin{proof}
Note that \(T_i(c) = T_i(X_i \,|\, c) = T(X_i^T \,|\, c) = T(0 \,|\, c - X_i^T) = T(0\,|\, \hat{c})\)
and \(D_i(c) = D(X_i - c \,|\, 0) = D(\hat{X}_i - \hat{c})\) with \(\hat{c} \defeq c - X_i^T \in \supp(T)\)
by applying equation (\ref{eq:transf_cutoff}) and Proposition \ref{prop:decomposition}.
\Halmos
\end{proof}
\propdisjoint*
\begin{proof}
It is easy to see that \(\alwayst(T, D) \cap \nevert(T,D) = \emptyset\)
and \(\comp(T, D) \cap \defier(T, D) = \emptyset\).
For \(i \in C \defeq \nevert(T, D) \cup \alwayst(T,D)\)
it follows that \(D(X_i^{\perp T} - c)\) is constant for all \(c \in \supp(T)\).
Suppose that \(i \in C \cap \comp(T, D)\).
The latter would mean, that \(T(0 \,|\, c)  = D(X_i^{\perp T} - c)\) for all \(c \in \supp(T)\).
This contradicts \(\supp(T) \neq \{0\}\).
Now suppose that \(i \in C \cap \defier(T, D)\). 
The latter would mean that \(T(0 \,|\, c) \neq D(X_i^{\perp T}  - c)\) for all \(c \in \supp(T)\),
which implies that \(T\) is constant on \(\supp(T)\).
This again contradicts \(\supp(T) \neq \{0\}\).
\Halmos
\end{proof}
From now on whenever we assume that \(D\) is a cutoff rule 
we suppose that \(D\) does not depend on \(I_{k,i}(c)\) with \(c \neq 0\),
that is \(T\) and \(D\) are synchronous regarding their cutoffs.
\propexhaustive*
\begin{proof}
To show the first implication we suppose that \(\dim(\supp(T)) > 1\) 
and construct a \(D\) such that \(\indes(T, \, D) \neq \emptyset\).
Since \(\dim(\supp(T)) > 1\) here exist \(l, k \in S(T)\) with \(l \neq k\).
Define 
\[ 
    D = \left(T \land I_k \land I_l\right)
      \lor \left(\overline{T} \land \overline{I_k} \land I_l\right)
      \lor \left(I_k \land \overline{I_l}\right)
\]
and let \(X_i\) be such that \(X_{k,i} = 0 = X_{l,i} \).
Then \(D(X_i \,|\, 0) = 0\) and \(D(X_i \,|\, e_k) = 1\).
This shows that \(i \notin \alwayst(T,\,D) \cup \nevert(T,\,D)\) since \(0, \, e_k \in \supp(T)\).
Further
\[ D(X_i \,|\, e_k + e_l) = T(X_i \,|\, e_k + e_l) \]
and
\[ D(X_i \,|\, e_l) = \overline{T}(X_i \,|\, \hat{c}) \neq T(X_i \,|\, e_l) .\]
Thus, \(i \in \indes(T,\,D)\) since \(e_l, \, e_k \in \supp(T)\).

For the other implication we suppose that \(\dim(\supp(T)) = 1\).
Since \(\supp(T) \simeq \Rn\) one has \(T_i = I_{k,i}\) with \(k \in S(T)\).
Further, suppose \(i\) is in neither of the mentioned sets. Then 
\(T_i(c) \neq D_i(c)\) and \(T_i(\hat{c}) = D_i(\hat{c})\) for some \(c, \, \hat{c} \in \supp(T)\).
Thus, \(\hat{c} = c + \lambda e_k\) for appropriate \(\lambda \in \Rn\).
If \(T_i(c) = T_i(\hat{c})\) then \(D_i(c) \neq D_i(\hat{c})\). 
Since \(D\) is a cutoff rule and \(I_{i,j}(c) = I_{i,j}(\hat{c})\) for 
\(j \neq k\) we conclude \(T_i(c) = I_{k,i}(c) \neq I_{k,i}(\hat{c}) = T_i(\hat{c})\).
This contradicts the assumption.
If otherwise \(T_i(c) \neq T_i(\hat{c})\) one has \(D_i(c) = D_i(\hat{c})\).
Which would imply that \(D_i(c)\) is constant for \(c \in \supp(T)\),
and thus \(i \in \alwayst(T, D) \cup \nevert(T, D)\).
\Halmos
\end{proof}
\subsection{Examples}\label{appendix:examples}
\paragraph{Example (OR-Rules)}\label{ex:or_rule}
Let \(D \defeq \bigvee_{j=1}^K I_{j}\) and \(T \defeq \bigvee_{j=1}^k I_{j} \) for \(k \in \Nn_{< K}\).
Then \(\supp(T) = \Rn^k \times \{0\}^{K-k}\).
The potential outcomes reduce to \(Y_i  = Y_i(0, 0)(1-D_i) + Y_i(0, 1)D_i(1-T_i) + Y(1, 1)T_i\).
Further we have the following unit categorizations:
\begin{enumerate}
  \item \(\comp(T,\, D) = \{i \,|\,  X_i \in \Rn^k \times \Rn^{K-k}_{\leq 0} \}\)

    Note that in this case the additional or-conditions are zero, reducing
    the rule \(D_i\) to \(T_i\).

  \item \(\alwayst(T,\, D) = \{ i \,|\, \exists k < j \leq K \, : \,  X_{j,i} > 0 \} = \{i \,|\, X_i \in \Rn^k \times (\Rn^{K-k} \setminus \Rn^{K-k}_{\leq 0}) \} \)

    As long as there is any additional or condition that is always true, cutoff changes in \(\supp(T)\)
    do not affect \(D\).

  \item \(\nevert(T,\, D) = \emptyset\)

    Choose \(c \in \supp(T)\) such that \(X_{i,j} > c\) for \(0 \leq j \leq k\), then \(D_i = 1\).

  \item \(\defier(T, \, D) = \emptyset\)
    
    Note that \(T_i(c) = 1\) implies \(D_i(c) = 1\) for all \(c \in \supp(T)\).

\end{enumerate}
\paragraph{Example (XOR dominant rule)}
Let \(D \defeq (I_1 \lor I_2) \land (\overline{I}_1 \lor \overline{I}_2) \) 
and \(T \defeq I_1\).
Then \(\supp(T) = \Rn \times \{0\}\).
The potential outcomes formulation does not simplify much:
\[
  Y_i  = Y_i(0, 0)(1-T_i)(1-I_{2,i})
  + Y_i(1, 0) T_i I_{2,i} 
  + Y_i(0, 1)(1-T_i) I_{2,i}
  + Y(1, 1)T_i(1-I_{2,i})
\]
Note if \(X_{2,i} > 0\), one has \(D_i(c) = \overline{T}_i(c)\)
and otherwise \(D_i(c) = T_i(c)\) for all \(c \in \supp(T)\).
Thus:
\begin{enumerate}
  \item \(\comp(T,\, D) = \{i \,|\,  X_{2,i} \leq 0 \}\)

  \item \(\alwayst(T,\, D) = \emptyset\)

  \item \(\nevert(T,\, D) = \emptyset\)

  \item \(\defier(T, \, D) = \{i \,|\, X_{2,i} > 0\}\)
\end{enumerate}

\paragraph{Nevertakers, Compliers, Alwaystakers}
We now investigate potential setup for cutoff rules that admit compliers, alwaystakers and nevertakers.
In accordance with the empirical part let \(T \defeq I_1 \square I_2\) for \(\square \in \{\lor, \land\}\) and \(G = I_k\) for \(k = 1, 2\) and further let \(D\) be a cutoff rule.
Thus, only the number of additional variables in \(D\) limit the possible unit categories.
That is, with one additional cutoff variable \(I_3\) one 
could realize each tuple of the different item types (e.g. defier and nevertakers).
A common assumption in causal literature is that there are no defying or indecisive units. 
Thus, to have a setting with each of the remaining three categories one would require 
one more variable \(I_4\).
This leaves a single choice open (nevertakers, alwaystakers or compliers)
for the remaining region in \(\Rn^4\).
As in case of the cautious operator we decide to opt for compliance.
Thus, we arrive at
\[ D = (T \land I_3) \lor (\overline{I}_3 \land I_4)\]
for a model of the actual treatment decision with:
\begin{align*}
  \comp(T, D) &= \{ i \,|\, X_{3,i} > 0 \} \\
  \alwayst(T, D) &= \{ i \,|\, X_{3,i} \leq 0, X_{4,i} > 0 \} \\
  \nevert(T, D) &= \{ i \,|\, X_{3,i} \leq 0, X_{4,i} \leq 0 \} \\
  \defier(T, D) &= \emptyset 
\end{align*}
For the case \(\square = \land\) and \(G \defeq I_1\) one has 
\begin{align*}
  \comp(G, D) &= \{ i \,|\, X_{3,i} > 0, X_{2, i} > 0 \} \\
  \alwayst(G, D) &= \{ i \,|\, X_{3,i} \leq 0, X_{4,i} > 0 \} \\
  \nevert(G, D) &= \{ i \,|\, X_{3,i} \leq 0, X_{4,i} \leq 0 \}
                  \cup \{ i \,|\, X_{2,i} \leq 0, X_{4,i} \leq 0 \}
                  \cup \{ i \,|\, X_{2,i} \leq 0, X_{3,i} > 0 \} \\
  \defier(G, D) &= \emptyset 
\end{align*}
That is \(X_4\) could be ignored for estimating the complier effect of \(G\) due to our specific choice,
whereas \(X_3\) and \(X_2\) are the switches required to identify \(\comp(G, D)\) completely.
Again \(X_2\) might be known from \(T\), which helps to reduce the observational dataset accordingly.

\subsection{Inheritance}\label{appendix:general_principles}
The following lemma is the basis for duality statements involving the unit categories.
\begin{lemma}\label{lem:supp_negate}
  The decomposition in Proposition \ref{prop:decomposition} is stable under negation,
  that is
  \[ X^T_i = X^{\overline{T}}_i \mbox{ and } X_i^{\perp T} = X_i^{\perp \overline{T}}\]
  In particular
  \(\supp(T) = \supp(\overline{T})\) and \(N^T = N^{\overline{T}}\).
\end{lemma}
\begin{proof}
  If \(T(0\,|\, c) \neq T(0\,|\, c + \lambda e_k)\) holds,
  the same is true for \(\overline{T}\) and vice versa.
  This shows that \(S(T) = S(\overline{T})\) and \(\supp(T) = \supp(\overline{T})\).
  Now let \(X \in \Rn^K\) such that \(T(X \,|\, c) = T(X \,|\, 0)\) for all \(c \in \Rn^K\).
  Then the same equality holds for \(\overline{T}\) and vice versa.
  Thus, \(N^T = N^{\overline{T}}\).
  Using that the decomposition in Proposition \ref{prop:decomposition} is a direct sum 
  of linear spaces we get
  \[
      X_i^T - X_i^{\overline{T}} 
      = X_i^{\perp \overline{T}} - X_i^{\perp T} \in \supp(T) \cap N^T
      = \{0\}
  \]
  which concludes the proof.
  \Halmos
\end{proof}
\propdualities*
\begin{proof}
Using Lemma \ref{lem:supp_negate} we have
\( D(X_i^{\perp T} - c) = D(X_i^{\perp \overline{T}} - c)\)
for \(X_i \in \Rn^K\) and \(c \in \supp(T) = \supp(\overline{T})\).
The equalities in of Definition \ref{def:unit_categories} are easily checked.
\Halmos
\end{proof}
\propboundatnt*
\begin{proof}
One has \((X_i^{\perp G})^{\perp T} = X_i^{\perp T}\).
To see this note that
\(X_i = (X_i^{\perp G})^T + X_i^G + (X_i^{\perp G})^{\perp T}\) with
\((X_i^{\perp G})^T + X_i^G \in \supp(T)\).
Then the claim follows by the uniqueness of the decomposition according to Proposition \ref{prop:decomposition} .
With this one obtains
\begin{align*}
    D(X_i^{\perp G} - c^G) = D\left((X_i^{\perp G})^{\perp T} + (X_i^{\perp G})^T - c^G\right) 
                           = D(X_i^{\perp T} - c)
\end{align*}
with \(c = c^G - (X_i^{\perp G})^T \in \supp(T)\)
which concludes the proof.
\Halmos
\end{proof}
\begin{lemma}\label{lem:supp_operation}
  Let \(G, \, H\) be cutoff rules on \(\Rn^K\) with \(\supp(T) = \supp(G) \oplus \supp(H)\).
  Then \(N^G = \supp(H) \oplus N^T\).
\end{lemma}
\begin{proof}
Let \(X \in \supp(H) \cap N^T = \{0\}\), then
\(X \in \supp(T)\) since \(\supp(H) \subset \supp(T)\) and thus \(X \in N^T \cap \supp(T) = \{0\}\).
This shows that the sum on the right side is indeed a direct sum, that is \(\supp(H) + N^T = \supp(H) \oplus N^T\).
Employing Proposition \ref{prop:decomposition} one has \(N^G \oplus \supp(G) = \Rn^K\). 
Using this and \(\supp(G) \subset \supp(T)\) it follows that \(N^T \subset N^G\).
Also note that \(\supp(H) \cap \supp(G) = \{0\}\) and thus \(\supp(H) \subset N^G\).
This shows \(\supp(H) \oplus N^T \subset N^G\) since \(N^G\) is a linear space.
Now let \(X \in N^G\). With Proposition \ref{prop:decomposition} one has 
\(X = X^{\perp T} = X^T + X^{\perp T}\) with \(X^T \in \supp(T)\) and \(X^{\perp T} \in N^T \subset N^G\)
and further \(X^T = X^G + X^H\) with \(X^G \in \supp(G)\) and \(X^H \in \supp(H) \subset N^G\).
Thus, \(X^G = X - X^H - X^{\perp T} \in \supp(G) \cap N^G\) and with this
\(X =  X^H + X^{\perp T} \in \supp(H) \oplus N^T\).
\Halmos
\end{proof}
\propoperationsimple*

\begin{proof}
We only show the case \(T = G \land H\) since the other follows in the same way.
Let \(i \in \comp(G, T)\). This is equivalent to 
\[G(X_i \,|\, c^G) = G(X_i \,|\, c^G) \land H(X_i \,|\, c^G) = G(X_i \,|\, c^G) \land H(X_i^H \,|\, 0) \quad \forall \, c^G \in \supp(G) \]
where the last equality follows from \(\supp(G) \subset N^H\)
(Proposition \ref{prop:decomposition}).
Since the support of \(G\) is non-trivial this is equivalent to \(H(X_i^H \,|\, 0) = 1\).
Now let \(i \in \nevert(G, T)\).
This is equivalent to:
\[0 = G(X_i \,|\, c^G) \land H(X_i^H \,|\, 0) \quad \forall \, c^G \in \supp(G) \]
Again since the support of \(G\) is non-trivial this is equivalent to \(H(X_i^H \,|\, 0) = 0\).
\Halmos
\end{proof}
\propboundcompdef*
\begin{proof}
Let \(T = G \land H\) and \(X_i \in \comp(T, D)\) then
\[
  G(0 \,|\, c^G) \land H(0 \,|\, c^H) = 
  G(0\,|\, c) \land H(0 \,|\, c) = 
  T(0 \,|\, c) = 
  D(X_i^{\perp T} - c)
\]
for all \(c = c^G + c^H \in \supp(G) \oplus \supp(H) = \supp(T)\).
The first equality follows from \(c^H \in N^G\) and \(c^G \in N^H\)
using Lemma \ref{lem:supp_operation}.
The above also holds for \(c = c^G - X^H_i\).
Further \(X_i^{\perp G} = X_i^H + X_i^{\perp T}\) according to Lemma \ref{lem:supp_operation}
and with this
\[ G(0 \,|\, c^G) \land H(X_i^H \,|\, 0) = D(X^{\perp G}_i - c^G) \]
for all \(c^G \in \supp(G)\).
If \(H(X_i^H \,|\, 0) = 1\) then it follows that \(i \in \comp(G, D)\).
Otherwise, \(D(X_i^{\perp G} - c^G) = 0\)
for all \(c^G \in \supp(G)\). Since \(\{0\} \neq \supp(G)\) one has \(i \in \nevert(G, D)\).
In summary, using Lemma \ref{prop:operation_simple} one derives
\begin{align*}
  \comp(G, T) \cap \comp(T, D) &\subset \comp(G, D) \\
  \nevert(G, T) \cap \comp(T, D) &\subset \nevert(G, D)
\end{align*}
and with it:
\[ \comp(T, D) \subset \comp(G, D) \cup \nevert(G, D) \]
Using the same argument one derives the analog result for \(T = G \lor H\).
\Halmos
\end{proof}
%
%
%
%

\section{Additional Numerical Results}\label{app:addres}
In this appendix we present additional results and supplementary material to the simulation results in Section \ref{sec:sim}.

\subsection{Simulation for Cautious Operator}
\paragraph{Cutoff $c_Y$}
\begin{table}[H]
\centering
\scalebox{0.5}{%
\begin{tabular}{llrrrrrrr}
\toprule
 &  & Mean Bias & s.e. & Coverage & RMSE left & Log loss left & RMSE right & Log loss right \\
setting & method &  &  &  &  &  &  &  \\
\midrule
\multirow[t]{6}{*}{Fuzzy} & RDD Conventional Covs & -0.0158 & 0.0124 & 0.6867 &  &  &  &  \\
 & RDD Without Covs & -0.0189 & 0.0117 & 0.5200 &  &  &  &  \\
 & RDFlex Boosting & 0.2453 & 7.8523 & 0.9467 & 0.0268 & 0.0313 & 0.0741 & 0.4072 \\
 & RDFlex Global Lasso & 0.1593 & 9.1373 & 0.9800 & 0.0374 & 0.0132 & 0.0707 & 0.5320 \\
 & RDFlex Lasso & -13.0991 & 690.4369 & 0.9600 & 0.0275 & 0.0202 & 0.0707 & 0.4681 \\
 & RDFlex Stacking & -0.3166 & 16.4743 & 0.9667 & 0.0257 & 0.0104 & 0.0732 & 0.3574 \\
\cline{1-9}
\multirow[t]{6}{*}{Fuzzy Subset} & RDD Conventional Covs & -0.0003 & 0.0085 & 0.9754 &  &  &  &  \\
 & RDD Without Covs & -0.0002 & 0.0085 & 0.9713 &  &  &  &  \\
 & RDFlex Boosting & -0.0002 & 0.0060 & 0.9713 & 0.0359 & 0.1205 & 0.0950 & 0.1473 \\
 & RDFlex Global Lasso & 0.0002 & 0.0057 & 0.9836 & 0.0336 & 0.1394 & 0.0936 & 0.1468 \\
 & RDFlex Lasso & 0.0002 & 0.0057 & 0.9795 & 0.0333 & 0.2688 & 0.0936 & 0.1580 \\
 & RDFlex Stacking & 0.0000 & 0.0058 & 0.9754 & 0.0337 & 0.0321 & 0.0938 & 0.1460 \\
\cline{1-9}
\multirow[t]{6}{*}{Sharp} & RDD Conventional Covs & -0.0043 & 0.0073 & 0.8720 &  &  &  &  \\
 & RDD Without Covs & -0.0050 & 0.0074 & 0.8880 &  &  &  &  \\
 & RDFlex Boosting & -0.0040 & 0.0059 & 0.8640 & 0.0257 &  & 0.0749 &  \\
 & RDFlex Global Lasso & -0.0049 & 0.0075 & 0.8840 & 0.0353 &  & 0.0728 &  \\
 & RDFlex Lasso & -0.0045 & 0.0077 & 0.8800 & 0.0271 &  & 0.0726 &  \\
 & RDFlex Stacking & -0.0048 & 0.0070 & 0.8680 & 0.0255 &  & 0.0732 &  \\
\cline{1-9}
\multirow[t]{6}{*}{Sharp Subset} & RDD Conventional Covs & -0.0005 & 0.0078 & 0.9800 &  &  &  &  \\
 & RDD Without Covs & -0.0004 & 0.0078 & 0.9720 &  &  &  &  \\
 & RDFlex Boosting & -0.0007 & 0.0057 & 0.9840 & 0.0361 &  & 0.0953 &  \\
 & RDFlex Global Lasso & -0.0003 & 0.0054 & 0.9800 & 0.0337 &  & 0.0940 &  \\
 & RDFlex Lasso & -0.0003 & 0.0054 & 0.9760 & 0.0333 &  & 0.0940 &  \\
 & RDFlex Stacking & -0.0004 & 0.0054 & 0.9720 & 0.0337 &  & 0.0939 &  \\
\cline{1-9}
\bottomrule
\end{tabular}

}
\caption{Mean Bias and standard error, coverage and first stage prediction quality of different estimators in estimation at $c_Y$.}\label{tab:yield}
\end{table}
\paragraph{Two-dimensional estimator}
\begin{table}[H]
\centering
\scalebox{0.5}{%
\begin{tabular}{llrrrrrrr}
\toprule
 &  & Mean Bias & s.e. & Coverage & RMSE left & Log loss left & RMSE right & Log loss right \\
setting & method &  &  &  &  &  &  &  \\
\midrule
\multirow[t]{6}{*}{2D Binding Fuzzy} & RDD Conventional Covs & -0.0032 & 0.0092 & 0.5080 &  &  &  &  \\
 & RDD Without Covs & -0.0020 & 0.0092 & 0.5080 &  &  &  &  \\
 & RDFlex Boosting & -0.0023 & 0.0073 & 0.5160 & 0.0292 & 0.0658 & 0.1008 & 0.2655 \\
 & RDFlex Global Lasso & -0.0029 & 0.0078 & 0.5120 & 0.0416 & 0.1128 & 0.0993 & 0.3470 \\
 & RDFlex Lasso & -0.0022 & 0.0074 & 0.5160 & 0.0288 & 0.3225 & 0.0987 & 0.4954 \\
 & RDFlex Stacking & -0.0023 & 0.0073 & 0.5120 & 0.0281 & 0.0155 & 0.0988 & 0.2438 \\
\cline{1-9}
\multirow[t]{6}{*}{2D Binding Sharp} & RDD Conventional Covs & -0.0021 & 0.0081 & 0.9600 &  &  &  &  \\
 & RDD Without Covs & -0.0008 & 0.0081 & 0.9480 &  &  &  &  \\
 & RDFlex Boosting & -0.0019 & 0.0069 & 0.9720 & 0.0298 &  & 0.1011 &  \\
 & RDFlex Global Lasso & -0.0018 & 0.0068 & 0.9520 & 0.0420 &  & 0.0991 &  \\
 & RDFlex Lasso & -0.0011 & 0.0069 & 0.9520 & 0.0288 &  & 0.0987 &  \\
 & RDFlex Stacking & -0.0014 & 0.0069 & 0.9480 & 0.0282 &  & 0.0989 &  \\
\cline{1-9}
\multirow[t]{6}{*}{2D Euclidean Fuzzy} & RDD Conventional Covs & -0.0226 & 0.0119 & 0.4360 &  &  &  &  \\
 & RDD Without Covs & -0.0229 & 0.0122 & 0.4640 &  &  &  &  \\
 & RDFlex Boosting & -0.0234 & 0.0094 & 0.4000 & 0.0255 & 0.0744 & 0.1054 & 0.0654 \\
 & RDFlex Global Lasso & -0.0208 & 0.0106 & 0.5920 & 0.0349 & 0.1123 & 0.1037 & 0.2470 \\
 & RDFlex Lasso & -0.0226 & 0.0094 & 0.4320 & 0.0250 & 0.4624 & 0.1020 & 0.4017 \\
 & RDFlex Stacking & -0.0223 & 0.0098 & 0.4640 & 0.0252 & 0.0179 & 0.1025 & 0.0154 \\
\cline{1-9}
\multirow[t]{6}{*}{2D Euclidean Sharp} & RDD Conventional Covs & -0.0255 & 0.0119 & 0.3120 &  &  &  &  \\
 & RDD Without Covs & -0.0257 & 0.0122 & 0.3160 &  &  &  &  \\
 & RDFlex Boosting & -0.0262 & 0.0104 & 0.3560 & 0.0254 &  & 0.1068 &  \\
 & RDFlex Global Lasso & -0.0253 & 0.0102 & 0.3440 & 0.0350 &  & 0.1038 &  \\
 & RDFlex Lasso & -0.0257 & 0.0103 & 0.3480 & 0.0247 &  & 0.1031 &  \\
 & RDFlex Stacking & -0.0258 & 0.0103 & 0.3200 & 0.0252 &  & 0.1034 &  \\
\cline{1-9}
\bottomrule
\end{tabular}

}
\caption{Mean Bias and standard error, coverage and first stage prediction quality of different estimators in estimation with two-dimensional estimators.}\label{tab:binding}
\end{table}

\subsection{Simulation for Acknowledging Operator}
\paragraph{Cutoff $c_D$}
\begin{table}[H]
\centering
\scalebox{0.5}{%
\begin{tabular}{llrrrrrrr}
\toprule
 &  & Mean Bias & s.e. & Coverage & RMSE left & Log loss left & RMSE right & Log loss right \\
setting & method &  &  &  &  &  &  &  \\
\midrule
\multirow[t]{6}{*}{Fuzzy} & RDD Conventional Covs & -0.0019 & 0.0102 & 0.9735 &  &  &  &  \\
 & RDD Without Covs & -0.0017 & 0.0103 & 0.9823 &  &  &  &  \\
 & RDFlex Boosting & -0.0018 & 0.0086 & 0.9912 & 0.0243 & 0.0580 & 0.0988 & 0.3837 \\
 & RDFlex Global Lasso & -0.0020 & 0.0086 & 1.0000 & 0.0395 & 0.1020 & 0.0964 & 0.5188 \\
 & RDFlex Lasso & -0.0021 & 0.0083 & 0.9823 & 0.0242 & 0.1130 & 0.0958 & 0.4996 \\
 & RDFlex Stacking & -0.0018 & 0.0085 & 0.9823 & 0.0243 & 0.0135 & 0.0960 & 0.3813 \\
\cline{1-9}
\multirow[t]{6}{*}{Sharp} & RDD Conventional Covs & -0.0042 & 0.0067 & 0.9320 &  &  &  &  \\
 & RDD Without Covs & -0.0042 & 0.0067 & 0.9360 &  &  &  &  \\
 & RDFlex Boosting & -0.0042 & 0.0058 & 0.9200 & 0.0243 &  & 0.0984 &  \\
 & RDFlex Global Lasso & -0.0042 & 0.0057 & 0.9360 & 0.0396 &  & 0.0956 &  \\
 & RDFlex Lasso & -0.0042 & 0.0057 & 0.9360 & 0.0241 &  & 0.0951 &  \\
 & RDFlex Stacking & -0.0042 & 0.0057 & 0.9320 & 0.0243 &  & 0.0954 &  \\
\cline{1-9}
\multirow[t]{6}{*}{Sharp Subset} & RDD Conventional Covs & 0.0010 & 0.0101 & 0.9720 &  &  &  &  \\
 & RDD Without Covs & 0.0009 & 0.0102 & 0.9720 &  &  &  &  \\
 & RDFlex Boosting & 0.0005 & 0.0089 & 0.9760 & 0.0261 &  & 0.1122 &  \\
 & RDFlex Global Lasso & 0.0009 & 0.0084 & 0.9720 & 0.0255 &  & 0.1083 &  \\
 & RDFlex Lasso & 0.0009 & 0.0085 & 0.9720 & 0.0247 &  & 0.1078 &  \\
 & RDFlex Stacking & 0.0009 & 0.0085 & 0.9720 & 0.0258 &  & 0.1080 &  \\
\cline{1-9}
\bottomrule
\end{tabular}

}
\caption{Mean Bias and standard error, coverage and first stage prediction quality of different estimators in estimation at $c_D$.}\label{tab:dis_nonvtk}
\end{table}
\paragraph{Cutoff $c_Y$}
\begin{table}[H]
\centering
\scalebox{0.5}{%
\begin{tabular}{llrrrrrrr}
\toprule
 &  & Mean Bias & s.e. & Coverage & RMSE left & Log loss left & RMSE right & Log loss right \\
setting & method &  &  &  &  &  &  &  \\
\midrule
\multirow[t]{6}{*}{Fuzzy} & RDD Conventional Covs & -0.0160 & 0.0125 & 0.6267 &  &  &  &  \\
 & RDD Without Covs & -0.0193 & 0.0115 & 0.5067 &  &  &  &  \\
 & RDFlex Boosting & -0.0390 & 0.4699 & 0.9800 & 0.0265 & 0.0338 & 0.0784 & 0.3830 \\
 & RDFlex Global Lasso & 0.0835 & 3.5647 & 0.9867 & 0.0380 & 0.0134 & 0.0747 & 0.5802 \\
 & RDFlex Lasso & -0.0697 & 4.5786 & 0.9867 & 0.0276 & 0.0213 & 0.0744 & 0.4486 \\
 & RDFlex Stacking & 0.0001 & 1.7755 & 0.9600 & 0.0257 & 0.0096 & 0.0777 & 0.3278 \\
\cline{1-9}
\multirow[t]{6}{*}{Sharp} & RDD Conventional Covs & -0.0035 & 0.0081 & 0.8480 &  &  &  &  \\
 & RDD Without Covs & -0.0044 & 0.0079 & 0.8600 &  &  &  &  \\
 & RDFlex Boosting & -0.0037 & 0.0069 & 0.8480 & 0.0257 &  & 0.0789 &  \\
 & RDFlex Global Lasso & -0.0040 & 0.0081 & 0.8600 & 0.0357 &  & 0.0768 &  \\
 & RDFlex Lasso & -0.0041 & 0.0086 & 0.8680 & 0.0271 &  & 0.0766 &  \\
 & RDFlex Stacking & -0.0042 & 0.0077 & 0.8600 & 0.0256 &  & 0.0771 &  \\
\cline{1-9}
\multirow[t]{6}{*}{Sharp Subset} & RDD Conventional Covs & -0.0008 & 0.0081 & 0.9760 &  &  &  &  \\
 & RDD Without Covs & -0.0008 & 0.0080 & 0.9760 &  &  &  &  \\
 & RDFlex Boosting & -0.0007 & 0.0059 & 0.9720 & 0.0366 &  & 0.0985 &  \\
 & RDFlex Global Lasso & -0.0006 & 0.0056 & 0.9800 & 0.0339 &  & 0.0970 &  \\
 & RDFlex Lasso & -0.0006 & 0.0056 & 0.9800 & 0.0336 &  & 0.0971 &  \\
 & RDFlex Stacking & -0.0007 & 0.0055 & 0.9760 & 0.0333 &  & 0.0970 &  \\
\cline{1-9}
\bottomrule
\end{tabular}

}
\caption{Mean Bias and standard error, coverage and first stage prediction quality of different estimators in estimation at $c_Y$.}\label{tab:yield_nonvtk}
\end{table}

\section*{Acknowledgements}
This work was funded by the Bavarian Joint Research Program (BayVFP) – Digitization (Funding reference: DIK0294/01)
and the German Research Foundation (DFG) (Grant no 448795504).
The research partners ams-OSRAM and Economic AI kindly thank the the VDI/VDE-IT Munich for the organization
and the Free State of Bavaria as well as the Federal Republic of Germany, for the financial support. 

\bibliographystyle{plainnat}
\bibliography{mybib.bib}

\begin{thebibliography}{44}
\providecommand{\natexlab}[1]{#1}
\providecommand{\url}[1]{\texttt{#1}}
\expandafter\ifx\csname urlstyle\endcsname\relax
  \providecommand{\doi}[1]{doi: #1}\else
  \providecommand{\doi}{doi: \begingroup \urlstyle{rm}\Url}\fi

\bibitem[An et~al.(2024)An, Branson, and Miratrix]{an2024}
Lily An, Zach Branson, and Luke Miratrix.
\newblock Using gaussian process regression in two-dimensional regression
  discontinuity designs. edworkingpaper no. 24-1043.
\newblock \emph{Annenberg Institute for School Reform at Brown University},
  2024.

\bibitem[Angrist and Lavy(1999)]{Angrist1999}
J.~D. Angrist and V.~Lavy.
\newblock Using maimonides’ rule to estimate the effect of class size on
  scholastic achievement.
\newblock \emph{The Quarterly Journal of Economics}, 114\penalty0 (2):\penalty0
  533–575, May 1999.
\newblock ISSN 1531-4650.
\newblock \doi{10.1162/003355399556061}.
\newblock URL \url{http://dx.doi.org/10.1162/003355399556061}.

\bibitem[Arai et~al.(2025)Arai, Otsu, and Seo]{Arai2025}
Yoichi Arai, Taisuke Otsu, and Myung~Hwan Seo.
\newblock Regression discontinuity design with potentially many covariates.
\newblock \emph{Econometric Theory}, page 1–36, February 2025.
\newblock ISSN 1469-4360.
\newblock \doi{10.1017/s0266466624000239}.
\newblock URL \url{http://dx.doi.org/10.1017/S0266466624000239}.

\bibitem[Austin(2011)]{Austin2011}
Peter~C Austin.
\newblock An introduction to propensity score methods for reducing the effects
  of confounding in observational studies.
\newblock \emph{Multivariate Behav. Res.}, 46\penalty0 (3):\penalty0 399--424,
  May 2011.

\bibitem[Black(1999)]{Black1999}
S~E Black.
\newblock Do better schools matter? parental valuation of elementary education.
\newblock \emph{Q. J. Econ.}, 114\penalty0 (2):\penalty0 577--599, May 1999.

\bibitem[Calonico et~al.(2014)Calonico, Cattaneo, and Titiunik]{Calonico2014}
Sebastian Calonico, Matias~D Cattaneo, and Rocio Titiunik.
\newblock Robust nonparametric confidence intervals for
  regression-discontinuity designs.
\newblock \emph{Econometrica}, 82\penalty0 (6):\penalty0 2295--2326, November
  2014.

\bibitem[Calonico et~al.(2019{\natexlab{a}})Calonico, Cattaneo, and
  Farrell]{Calonico2019b}
Sebastian Calonico, Matias~D Cattaneo, and Max~H Farrell.
\newblock Optimal bandwidth choice for robust bias-corrected inference in
  regression discontinuity designs.
\newblock \emph{The Econometrics Journal}, 23\penalty0 (2):\penalty0 192–210,
  November 2019{\natexlab{a}}.
\newblock ISSN 1368-423X.
\newblock \doi{10.1093/ectj/utz022}.
\newblock URL \url{http://dx.doi.org/10.1093/ectj/utz022}.

\bibitem[Calonico et~al.(2019{\natexlab{b}})Calonico, Cattaneo, Farrell, and
  Titiunik]{calonico2019}
Sebastian Calonico, Matias~D. Cattaneo, Max~H. Farrell, and Rocío Titiunik.
\newblock Regression discontinuity designs using covariates.
\newblock \emph{The Review of Economics and Statistics}, 101\penalty0
  (3):\penalty0 442--451, 07 2019{\natexlab{b}}.
\newblock ISSN 0034-6535.
\newblock \doi{10.1162/rest_a_00760}.
\newblock URL \url{https://doi.org/10.1162/rest\_a\_00760}.

\bibitem[Calvo et~al.(2019)Calvo, Cui, and Serpa]{Calvo2019}
Eduard Calvo, Ruomeng Cui, and Juan~Camilo Serpa.
\newblock Oversight and efficiency in public projects: A regression
  discontinuity analysis.
\newblock \emph{Manage. Sci.}, 65\penalty0 (12):\penalty0 5651--5675, December
  2019.

\bibitem[Card et~al.(2015)Card, Lee, Pei, and Weber]{Card2015}
David Card, David~S. Lee, Zhuan Pei, and Andrea Weber.
\newblock Inference on causal effects in a generalized regression kink design.
\newblock \emph{Econometrica}, 83\penalty0 (6):\penalty0 2453–2483, 2015.
\newblock ISSN 0012-9682.
\newblock \doi{10.3982/ecta11224}.
\newblock URL \url{http://dx.doi.org/10.3982/ECTA11224}.

\bibitem[Cattaneo et~al.(2019)Cattaneo, Idrobo, and Titiunik]{Cattaneo2019}
Matias~D. Cattaneo, Nicolas Idrobo, and Rocio Titiunik.
\newblock A practical introduction to regression discontinuity designs:
  Foundations.
\newblock November 2019.
\newblock \doi{10.1017/9781108684606}.

\bibitem[Cattaneo et~al.(2025)Cattaneo, Titiunik, and Yu]{Cattaneo2025}
Matias~D. Cattaneo, Rocio Titiunik, and Ruiqi~Rae Yu.
\newblock Estimation and inference in boundary discontinuity designs, 2025.
\newblock URL \url{https://arxiv.org/abs/2505.05670}.

\bibitem[Chernozhukov et~al.(2024)Chernozhukov, Hansen, Kallus, Spindler, and
  Syrgkanis]{chernozhukov2024}
Victor Chernozhukov, Christian Hansen, Nathan Kallus, Martin Spindler, and
  Vasilis Syrgkanis.
\newblock Applied causal inference powered by ml and ai, 2024.
\newblock URL \url{https://arxiv.org/abs/2403.02467}.

\bibitem[Cho et~al.(2017)Cho, Park, Kim, and Schubert]{Cho2017}
Jaehee Cho, Jun~Hyuk Park, Jong~Kyu Kim, and E.~Fred Schubert.
\newblock White light‐emitting diodes: History, progress, and future.
\newblock \emph{Laser \& Photonics Reviews}, 11, 2017.

\bibitem[Choi and Lee(2018)]{Choi2018}
Jin-young Choi and Myoung-jae Lee.
\newblock Regression discontinuity with multiple running variables allowing
  partial effects.
\newblock \emph{Political Analysis}, 26\penalty0 (3):\penalty0 258–274, July
  2018.
\newblock ISSN 1476-4989.
\newblock \doi{10.1017/pan.2018.13}.
\newblock URL \url{http://dx.doi.org/10.1017/pan.2018.13}.

\bibitem[Choi and Lee(2023)]{Choi2023}
Jin-young Choi and Myoung-jae Lee.
\newblock Complier and monotonicity for fuzzy multi-score regression
  discontinuity with partial effects.
\newblock \emph{Economics Letters}, 228:\penalty0 111169, 2023.
\newblock ISSN 0165-1765.
\newblock \doi{https://doi.org/10.1016/j.econlet.2023.111169}.
\newblock URL
  \url{https://www.sciencedirect.com/science/article/pii/S0165176523001945}.

\bibitem[Flammer(2015)]{Flammer2015}
Caroline Flammer.
\newblock Does corporate social responsibility lead to superior financial
  performance? a regression discontinuity approach.
\newblock \emph{Management Science}, 61\penalty0 (11):\penalty0 2549–2568,
  November 2015.
\newblock ISSN 1526-5501.
\newblock \doi{10.1287/mnsc.2014.2038}.
\newblock URL \url{http://dx.doi.org/10.1287/mnsc.2014.2038}.

\bibitem[Hahn et~al.(2001)Hahn, Todd, and der Klaauw]{Hahn2001}
Jinyong Hahn, Petra Todd, and Wilbert~Van der Klaauw.
\newblock Identification and estimation of treatment effects with a
  regression-discontinuity design.
\newblock \emph{Econometrica}, 69\penalty0 (1):\penalty0 201--209, 2001.
\newblock ISSN 00129682, 14680262.
\newblock URL \url{http://www.jstor.org/stable/2692190}.

\bibitem[Hartmann et~al.(2011)Hartmann, Nair, and Narayanan]{Hartmann2011}
Wesley Hartmann, Harikesh~S. Nair, and Sridhar Narayanan.
\newblock Identifying causal marketing mix effects using a regression
  discontinuity design.
\newblock \emph{Marketing Science}, 30\penalty0 (6):\penalty0 1079–1097,
  November 2011.
\newblock ISSN 1526-548X.
\newblock \doi{10.1287/mksc.1110.0670}.
\newblock URL \url{http://dx.doi.org/10.1287/mksc.1110.0670}.

\bibitem[Ho et~al.(2017)Ho, Lim, Reza, and Xia]{Ho2017}
Teck-Hua Ho, Noah Lim, Sadat Reza, and Xiaoyu Xia.
\newblock Om forum—causal inference models in operations management.
\newblock \emph{Manufacturing \& Service Operations Management}, 19\penalty0
  (4):\penalty0 509--525, 2017.
\newblock \doi{10.1287/msom.2017.0659}.
\newblock URL \url{https://doi.org/10.1287/msom.2017.0659}.

\bibitem[Hu et~al.(2021)Hu, Chopra, and Chen]{Hu2021}
Kejia Hu, Sunil Chopra, and Yuche Chen.
\newblock The effect of tightening standards on automakers’ non-compliance.
\newblock \emph{Production and Operations Management}, 30\penalty0
  (9):\penalty0 3094--3115, 2021.

\bibitem[H\"{u}nermund et~al.(2021)H\"{u}nermund, Kaminski, and
  Schmitt]{Hnermund2021}
Paul H\"{u}nermund, Jermain Kaminski, and Carla Schmitt.
\newblock Causal machine learning and business decision making.
\newblock \emph{SSRN Electronic Journal}, 2021.
\newblock ISSN 1556-5068.
\newblock \doi{10.2139/ssrn.3867326}.
\newblock URL \url{http://dx.doi.org/10.2139/ssrn.3867326}.

\bibitem[Imbens and Wager(2019)]{Imbens2019}
Guido Imbens and Stefan Wager.
\newblock Optimized regression discontinuity designs.
\newblock \emph{The Review of Economics and Statistics}, 101\penalty0
  (2):\penalty0 264–278, May 2019.
\newblock ISSN 1530-9142.
\newblock \doi{10.1162/rest_a_00793}.
\newblock URL \url{http://dx.doi.org/10.1162/rest_a_00793}.

\bibitem[Imbens and Zajonc(2009)]{Imbens2009}
Guido Imbens and Tristan Zajonc.
\newblock Regression discontinuity design with vector-argument assignment
  rules.
\newblock \emph{Unpublished paper}, page~8, 2009.

\bibitem[Imbens and Angrist(1994)]{Imbens1994}
Guido~W. Imbens and Joshua~D. Angrist.
\newblock Identification and estimation of local average treatment effects.
\newblock \emph{Econometrica}, 62\penalty0 (2):\penalty0 467, March 1994.
\newblock ISSN 0012-9682.
\newblock \doi{10.2307/2951620}.

\bibitem[Imbens and Lemieux(2008)]{Imbens2008}
Guido~W. Imbens and Thomas Lemieux.
\newblock Regression discontinuity designs: A guide to practice.
\newblock \emph{Journal of Econometrics}, 142\penalty0 (2):\penalty0 615–635,
  February 2008.
\newblock ISSN 0304-4076.
\newblock \doi{10.1016/j.jeconom.2007.05.001}.
\newblock URL \url{http://dx.doi.org/10.1016/j.jeconom.2007.05.001}.

\bibitem[Keele and Titiunik(2015)]{keele2015}
Luke~J Keele and Rocio Titiunik.
\newblock Geographic boundaries as regression discontinuities.
\newblock \emph{Political Analysis}, 23\penalty0 (1):\penalty0 127--155, 2015.

\bibitem[Kreiss and Rothe(2022)]{Kreiss2022}
Alexander Kreiss and Christoph Rothe.
\newblock Inference in regression discontinuity designs with high-dimensional
  covariates.
\newblock \emph{The Econometrics Journal}, 26\penalty0 (2):\penalty0 105–123,
  December 2022.
\newblock ISSN 1368-423X.
\newblock \doi{10.1093/ectj/utac029}.
\newblock URL \url{http://dx.doi.org/10.1093/ectj/utac029}.

\bibitem[Lee(2008)]{Lee2008}
David~S. Lee.
\newblock Randomized experiments from non-random selection in u.s. house
  elections.
\newblock \emph{Journal of Econometrics}, 142\penalty0 (2):\penalty0 675–697,
  February 2008.
\newblock ISSN 0304-4076.
\newblock \doi{10.1016/j.jeconom.2007.05.004}.
\newblock URL \url{http://dx.doi.org/10.1016/j.jeconom.2007.05.004}.

\bibitem[Lee and Lemieux(2010)]{Lee2010}
David~S Lee and Thomas Lemieux.
\newblock Regression discontinuity designs in economics.
\newblock \emph{Journal of Economic Literature}, 48\penalty0 (2):\penalty0
  281–355, June 2010.
\newblock ISSN 0022-0515.
\newblock \doi{10.1257/jel.48.2.281}.
\newblock URL \url{http://dx.doi.org/10.1257/jel.48.2.281}.

\bibitem[Leung et~al.(2020)Leung, Li, and Sun]{Leung2020}
Woon~Sau Leung, Jing Li, and Jiong Sun.
\newblock Labor unionization and supply-chain partners’ performance.
\newblock \emph{Production and Operations Management}, 29\penalty0
  (5):\penalty0 1325--1353, 2020.

\bibitem[Liu and Qi(2024)]{liu2024}
Yiqi Liu and Yuan Qi.
\newblock Using forests in multivariate regression discontinuity designs, 2024.
\newblock URL \url{https://arxiv.org/abs/2303.11721}.

\bibitem[Mithas et~al.(2022)Mithas, Chen, Lin, and
  De~Oliveira~Silveira]{Mithas2022}
Sunil Mithas, Yanzhen Chen, Yatang Lin, and Alysson De~Oliveira~Silveira.
\newblock On the causality and plausibility of treatment effects in operations
  management research.
\newblock \emph{Production and Operations Management}, 31\penalty0
  (12):\penalty0 4558–4571, December 2022.
\newblock ISSN 1937-5956.
\newblock \doi{10.1111/poms.13863}.
\newblock URL \url{http://dx.doi.org/10.1111/poms.13863}.

\bibitem[Noack et~al.(2024)Noack, Olma, and Rothe]{noack2024}
Claudia Noack, Tomasz Olma, and Christoph Rothe.
\newblock Flexible covariate adjustments in regression discontinuity designs,
  2024.
\newblock \url{https://arxiv.org/abs/2107.07942}.

\bibitem[Papay et~al.(2011)Papay, Willett, and Murnane]{papay2011}
John~P. Papay, John~B. Willett, and Richard~J. Murnane.
\newblock Extending the regression-discontinuity approach to multiple
  assignment variables.
\newblock \emph{Journal of Econometrics}, 161\penalty0 (2):\penalty0 203--207,
  2011.
\newblock ISSN 0304-4076.
\newblock \doi{https://doi.org/10.1016/j.jeconom.2010.12.008}.
\newblock URL
  \url{https://www.sciencedirect.com/science/article/pii/S0304407610002538}.

\bibitem[Porter et~al.(2017)Porter, Reardon, Unlu, Bloom, and
  Cimpian]{porter2017}
Kristin~E Porter, Sean~F Reardon, Fatih Unlu, Howard~S Bloom, and Joseph~R
  Cimpian.
\newblock Estimating causal effects of education interventions using a
  two-rating regression discontinuity design: Lessons from a simulation study
  and an application.
\newblock \emph{Journal of Research on Educational Effectiveness}, 10\penalty0
  (1):\penalty0 138--167, 2017.

\bibitem[Reardon and Robinson(2012)]{Reardon2012}
Sean~F. Reardon and Joseph~P. Robinson.
\newblock Regression discontinuity designs with multiple rating-score
  variables.
\newblock \emph{Journal of Research on Educational Effectiveness}, 5\penalty0
  (1):\penalty0 83–104, January 2012.
\newblock ISSN 1934-5739.
\newblock \doi{10.1080/19345747.2011.609583}.
\newblock URL \url{http://dx.doi.org/10.1080/19345747.2011.609583}.

\bibitem[Rubin(1974)]{Rubin1974}
Donald~B. Rubin.
\newblock Estimating causal effects of treatments in randomized and
  nonrandomized studies.
\newblock \emph{Journal of Educational Psychology}, 66\penalty0 (5):\penalty0
  688–701, October 1974.
\newblock ISSN 0022-0663.
\newblock \doi{10.1037/h0037350}.
\newblock URL \url{http://dx.doi.org/10.1037/h0037350}.

\bibitem[Rubin(2005)]{rubin2005}
Donald~B. Rubin.
\newblock Causal inference using potential outcomes: Design, modeling,
  decisions.
\newblock \emph{Journal of the American Statistical Association}, 100\penalty0
  (469):\penalty0 322--331, 2005.
\newblock ISSN 01621459.
\newblock URL \url{http://www.jstor.org/stable/27590541}.

\bibitem[Sabaei et~al.(2015)Sabaei, Erkoyuncu, and Roy]{Sabaei2015}
Davood Sabaei, John Erkoyuncu, and Rajkumar Roy.
\newblock A review of multi-criteria decision making methods for enhanced
  maintenance delivery.
\newblock \emph{Procedia CIRP}, 37:\penalty0 30--35, 2015.

\bibitem[Sawada et~al.(2025)Sawada, Ishihara, Kurisu, and Matsuda]{Sawada2025}
Masayuki Sawada, Takuya Ishihara, Daisuke Kurisu, and Yasumasa Matsuda.
\newblock Local-polynomial estimation for multivariate regression discontinuity
  designs, 2025.
\newblock URL \url{https://arxiv.org/abs/2402.08941}.

\bibitem[Schwarz et~al.(2024)Schwarz, Schacht, Klaassen, Grünbaum, Imhof, and
  Spindler]{schwarz2024}
Philipp Schwarz, Oliver Schacht, Sven Klaassen, Daniel Grünbaum, Sebastian
  Imhof, and Martin Spindler.
\newblock Management decisions in manufacturing using causal machine learning
  -- to rework, or not to rework?, 2024.
\newblock \url{https://arxiv.org/abs/2406.11308}.

\bibitem[Thistlethwaite and Campbell(1960)]{thistlethwaite1960}
Donald~L Thistlethwaite and Donald~T Campbell.
\newblock Regression-discontinuity analysis: An alternative to the ex post
  facto experiment.
\newblock \emph{J. Educ. Psychol.}, 51\penalty0 (6):\penalty0 309--317,
  December 1960.

\bibitem[Wong et~al.(2013)Wong, Steiner, and Cook]{wong2013}
Vivian~C Wong, Peter~M Steiner, and Thomas~D Cook.
\newblock Analyzing regression-discontinuity designs with multiple assignment
  variables: A comparative study of four estimation methods.
\newblock \emph{Journal of Educational and Behavioral Statistics}, 38\penalty0
  (2):\penalty0 107--141, 2013.

\end{thebibliography}

\end{document}